\documentclass[mnsc,nonblindrev]{informs3_SSRN}

\OneAndAHalfSpacedXI % current default line spacing
\usepackage{natbib}
\def\bibfont{\small}%
\TheoremsNumberedThrough     % Preferred (Theorem 1, Lemma 1, Theorem 2)
\ECRepeatTheorems

\EquationsNumberedThrough    % Default: (1), (2), ...
\MANUSCRIPTNO{MS-0000-0000.00}

\usepackage{color}
\usepackage{multimedia}
\usepackage{multicol}
\usepackage{bbm}
\usepackage{graphics}
\usepackage{graphicx}
\usepackage{booktabs}
\usepackage{romannum}

\usepackage{grffile}
\usepackage{float}
\usepackage{url}
\usepackage{mathabx}
\usepackage{supertabular}
\usepackage{rotating}
\usepackage{pdflscape}
\usepackage{verbatim}
\usepackage{listings}
\usepackage{xcolor}
\usepackage{multicol}
\usepackage[shortlabels]{enumitem}
\usepackage{caption,subcaption}
\usepackage{epstopdf}
\usepackage{algorithm}
\usepackage{algorithmic}
\usepackage{enumitem}
\usepackage{pifont}

\newcommand{\dd}{\mathrm{d}}
\newcommand{\one}{\mathbbm{1}}
\newcommand{\E}{\mathbb{E}}

\newcommand{\R}{{\mathbb  R}}

\newcommand{\f}{\mathcal{F}}

\newcommand{\bm}[1]{\boldsymbol{#1}}

\graphicspath{{Figures/}}

\usepackage[normalem]{ulem}
\begin{document}
	%%%%%%%%%%%%%%%%
	\pagenumbering{arabic}
	%\setcounter{page}{1}
	%\DOI{10.1287/xxxx.0000.0000}%
	
	\RUNAUTHOR{Dai et al.}
	
	\RUNTITLE{Merton's Problem with Recursive Perturbed Utility}
	
	% \TITLE{Bundling Information Goods of Decreasing Value}
	\TITLE{Merton's Problem with Recursive Perturbed Utility}
	
	\ARTICLEAUTHORS{%
		\AUTHOR{Min Dai}
		\AFF{Department of Applied Mathematics and School of Accounting and Finance, The Hong Kong Polytechnic University, Hung Hom, Kowloon, Hong Kong. \EMAIL{mindai@polyu.edu.hk}} %, \URL{}}
	\AUTHOR{Yuchao Dong}
	\AFF{Key Laboratory of Intelligent Computing and Applications (Tongji University), Ministry of Education and School of Mathematical Sciences, Tongji University, Shanghai 200092, China. \EMAIL{ycdong@tongji.edu.cn}}
	\AUTHOR{Yanwei Jia}
	\AFF{Department of Systems Engineering and Engineering Management, The Chinese University of Hong Kong, Shatin, New Territories, Hong Kong. \EMAIL{yanweijia@cuhk.edu.hk}}
	\AUTHOR{Xun Yu Zhou}
	\AFF{Department of Industrial Engineering and Operations Research \& Data Science Institute, Columbia University, New York, NY 10027, USA. \EMAIL{xz2574@columbia.edu}}
} % end of the block

\ABSTRACT{%
The classical Merton investment problem predicts deterministic, state-dependent portfolio {\it rules}; however, laboratory and field evidence suggests that individuals often prefer randomized decisions leading to stochastic and noisy choices. \citet{fudenberg2015stochastic} develop the additive perturbed utility theory to explain the preference for randomization in the static setting, which, however, becomes ill-posed or intractable in the dynamic setting.
We introduce the {\it recursive} perturbed utility (RPU), a special stochastic differential utility that incorporates an entropy-based preference for randomization into a recursive aggregator. RPU endogenizes the intertemporal trade-off between utilities from randomization and bequest via a discounting term dependent on past accumulated randomization, thereby avoiding excessive randomization and yielding a well-posed problem. In a general Markovian incomplete market with CRRA preferences, we prove that the RPU-optimal portfolio policy (in terms of the risk exposure ratio) is Gaussian and can be expressed in closed form, independent of wealth.  Its variance is inversely proportional to risk aversion and stock volatility, while its mean is based on the solution to a partial differential equation. Moreover, the mean is the sum of a myopic term and an intertemporal hedging term (against market incompleteness) that intertwines with policy randomization.
Finally, we carry out an asymptotic expansion in terms of the perturbed  utility weight to show  that the optimal mean policy deviates from the classical Merton policy at first order, while the associated relative wealth loss is of a higher order, quantifying  the financial cost of the preference for randomization.
}%

%\KEYWORDS{deterministic inventory theory; infinite linear programming duality;
	%  existence of optimal policies; semi-Markov decision process; cyclic schedule}

%\HISTORY{This paper was first submitted on 8 August 2020.}

\maketitle
%%%%%%%%%%%%%%%%%%%%%%%%%%%%%%%%%%%%%%%%%%%%%%%%%%%%%%%%%%%%%%%%%%%%%%

% Samples of sectioning (and lab\textbf{}eling) in MNSC
% NOTE: (1) \section and \subsection do NOT end with a period
%       (2) \subsubsection and lower need end punctuation
%       (3) capitalization is as shown (title style).
%
%\section{Introduction.}\label{intro} %%1.
%\subsection{Duality and the Classical EOQ Problem.}\label{class-EOQ} %% 1.1.
%\subsection{Outline.}\label{outline1} %% 1.2.
%\subsubsection{Cyclic Schedules for the General Deterministic SMDP.}
%  \label{cyclic-schedules} %% 1.2.1
%\section{Problem Description.}\label{problemdescription} %% 2.

% Text of your paper here
\section{Introduction}
\label{intro}

Merton’s continuous-time investment problem \citep{merton1969lifetime} is a cornerstone of intertemporal portfolio choice. It provides a tractable framework for characterizing optimal investment decisions and has been extended to a variety of preference specifications with rich economic implications. A common conclusion arising from these models is that a rational, risk-averse investor follows a {\it deterministic} function of time and state, also called a portfolio policy or rule.\footnote{Here ``state" includes {\it observable} variables including  portfolio worth, known factors and stock prices.}  Thus, one can predict investor behavior deterministically based on these models. However, experimental and empirical evidence suggests that people sometimes favor or even crave ``randomized" and noisy choices, with examples ranging from sushi omakase\footnote{``Omakase", meaning in Japanese ``to entrust", originally refers to a set sushi course whose contents are completely determined by the chef based on his/her observation of an individual customer.  Therefore, the surprise from the revelation of each sushi dish is a major fun part of the dining experience. Nowadays, omakase meals popular in many countries are often pre-announced set menus, thereby losing their original feature (and excitement) of randomization.}  and ``blind boxes"\footnote{Also called ``mystery boxes", blind boxes are curated, randomized packages sold for a fixed price, containing undisclosed items ranging from electronics and collectibles to clothing and snack boxes. The ultra-popular Labubu figures produced by Pop Mart are such an example.}  to portfolio allocations. In other words, the same individual does not necessarily repeat the same choice when faced with the same problem. This leads to {\it stochastic} choice and randomized policies. \citet{fudenberg2015stochastic} propose and develop the \emph{additive perturbed utility} (APU) for \emph{static} choice problems, which is expected utility plus a perturbation over choice probabilities. This perturbation captures an inherent desire to randomize, consistent also with the behavioral literature \citep{hey1995stochastic,dwenger2018flipping,miao2018probabilistic,permana2020people,agranov2017stochastic,agranov2025ranges}.
The most widely used perturbation is the \emph{entropy} function, which produces the logit/Luce stochastic choice when the choice set is finite \citep{luce1959individual,anderson1992discrete}.\footnote{For finite choice sets, APU connects to discrete-choice and random-utility models \citep{mcfadden1974measurement,mcfadden2001economic,anderson1992discrete,berry1994estimating,machina1985stochastic,mattsson2002probabilistic,feng2017relation} and has been extended to multi-period settings \citep{hotz1993conditional,hotz1994simulation}.}

%Motivated by this discrepancy, we introduce the notion of \emph{recursive perturbed utility} (RPU). We show that even under classical constant-relative-risk-aversion (CRRA) preferences, maximizing RPU in a general incomplete market yields an \emph{optimal stochastic} portfolio policy. We further contrast its implications for portfolio choice with the standard Merton benchmark without perturbations.

In this paper, we
 %We adopt the entropy perturbation and embed it in a dynamic portfolio problem where the control is continuum-valued. In our interpretation, randomizing the portfolio does not generate direct monetary payoffs; rather, the perturbation term is a non-monetary component of preferences that coexists with the CRRA bequest utility over terminal wealth. Because portfolio choices are made continuously over time, the agent accrues a \emph{flow} of utility from randomization---analogous to the flow utility from consumption in the classical Merton consumption–investment problem. To maintain this analogy clearly, we abstract from consumption and focus solely on portfolio choice.
extend APU to the {\it dynamic} setting inherent for Merton's problem. %, with entropy perturbation. %, analogous to the additive flow utility from consumption.
Now that  the control is continuum-valued, this extension from static to dynamic is by no means straightforward. Indeed, randomizing portfolios does not generate direct monetary payoffs; rather, the perturbation term is a non-monetary component of the preference in addition to the bequest utility (of the terminal wealth). Because portfolio choices are made continuously over time, the agent accrues a \emph{flow} of utility from randomization, analogous to the flow utility from consumption in the classical Merton consumption-investment problem.\footnote{In this paper we do not consider consumption so as to stay focused on the randomization utility flow. }
With this in mind, our first result shows that the dynamic counterpart of APU in the Merton setting  is {\it ill-posed} when risk aversion is  not sufficiently strong: the agent can drive the payoff to infinity by increasing the choice variance without bound. Two forces, specific to the dynamic setting, are behind this pathology. First, with a continuous control, the entropy-based randomization can deliver unbounded perturbation payoffs, unlike the bounded entropy in finite choice sets. Second, the intertemporal trade-off between randomization and investment fundamentally differs from the familiar trade-off between consumption and investment. Consumption sacrifices future wealth growth; randomization, by contrast, leaves the \emph{expected} growth rate unchanged while increasing volatility. Risk aversion, if too low, is therefore insufficient to deter excessive randomization. Moreover, even when risk aversion is sufficiently large to restore well-posedness, the additive model fails to yield tractable solutions—even in the Black-Scholes setting with constant-relative-risk-aversion (CRRA) preferences—making it difficult to characterize and analyze optimal policies.

These failures highlight more nuanced intertemporal interactions between utilities from randomization and from investment than the APU can capture. We therefore propose the \emph{recursive perturbed utility} (RPU), inspired by recursive utility developed by \citet{epstein1989substitution} in discrete time and by \citet{duffie1992stochastic} in continuous time. While recursive utility was introduced to separate risk aversion from the elasticity of intertemporal substitution in consumption, our goal is to \emph{endogenize} the effect of the entropy perturbation. The form of RPU resembles the Uzawa utility \citep{uzawa1968time}, which captures habit formation through an endogenous discount rate affected by past consumption.\footnote{See \citet{bergman1985time} for economic and financial implications of Uzawa preferences.} 
In our setting, past randomization influences the agent’s current preference in the form of time preference: the perturbation induces \emph{endogenous discounting} that dampens the flow of utility from further randomization as the latter accumulates over time. In other words, RPU gives rise to a history-dependent time preference
on randomization, such that the more or longer the investor has randomized  in the past, the less weight
she places on current randomization.

We study Merton's problem in an incomplete Markovian stock market  with exogenous stochastic factors (e.g., \citealt{wachter2002portfolio,liu2007portfolio,chacko2005dynamic}) and CRRA utilities within the RPU framework featuring  entropy perturbations. We prove that the optimal portfolio (in terms of the risk exposure ratio) policy must follow a Gaussian distribution that is dependent on the stochastic factors but independent of wealth. The variance of this policy admits a closed-form expression, which is inversely proportional to both relative risk aversion and  the instantaneous variance of stock returns. Thus, higher risk aversion or higher market volatility reduces choice noise. Meanwhile, the optimal mean is dependent only on time and factor, characterized by a partial differential equation (PDE). In general, the mean decomposes into a myopic term (that depends on risk aversion and the stock's instantaneous return-risk trade-off but is \emph{independent} of the randomization) and an intertemporal hedging term (that captures the hedging need due to market incompleteness and its correlation with the randomization). In complete markets or when the stock and factor move independently, the hedging term vanishes and the mean optimal policy is purely myopic, coinciding with the classical Merton solution. A notable special case is log utility (unit relative risk aversion), in which intertemporal hedging is absent, and RPU collapses to APU yielding a mean optimal policy identical to the Merton benchmark.\footnote{See \citet{ddj2020calibrating} for a related result in the context of entropy-regularized reinforcement learning with log utility.}

In general, however, RPU biases the optimal mean relative to the classical Merton benchmark (without randomization). We quantify this effect via an asymptotic expansion in the perturbation weight (called the “\emph{temperature}”). The deviation of the optimal mean from the classical Merton policy is first order in temperature, while the associated relative wealth loss from adopting the RPU-optimal policy (versus the classical benchmark) is of higher order. This in turn provides a transparent measure of the financial cost associated with the preference for randomized choices.

%\subsection*{Related Literature}

It is worth noting that the perturbed utility with entropy function has been widely used in the reinforcement learning (RL) literature, albeit with a different name ``\textit{entropy regularization}",  as seen in works such as \citet{ziebart2008maximum,haarnoja2018soft,zhao2019maximum} for the discrete-time setting and \cite{wang2018exploration} for the continuous-time one, among many others. The resulting optimal policies take the form of a Gibbs measure. However, the purposes of introducing randomized  policies and entropy regularization for RL are conceptually and fundamentally different: the aim  is to design  algorithms to solve decision-making problems in a data-driven way, often without knowledge of the environment’s probabilistic structure. The entropy term is added to the reward function as an explicit incentive to encourage {\it exploration} by randomizing choices, with the ultimate goal to balance exploration (learning) and exploitation (optimization).

By contrast, our paper follows \citet{fudenberg2015stochastic} and employs  entropy to capture  humans’ intrinsic preference for randomization.\footnote{\citet{fudenberg2015stochastic} discuss also APU as a way to tackle payoff uncertainty; see Section 5 therein and the related references cited. However, no model parameter  unavailability is involved.}  We study and characterize optimal behaviors under this perturbed utility paradigm by assuming the investor is rational and possesses full knowledge of the market model. Here, stochastic choices do not arise from the need to explore, but rather from a genuine preference for randomization.\footnote{A companion paper \citep{daidata} studies the Merton problem from the RL and computational perspective. It introduces an auxiliary problem with a class of Gaussian policies solvable by RL algorithms and proves that its optimal policy can be used to recover the optimal (deterministic) policy to the original Merton problem. Additionally, several other works take this exploratory framework to devise RL algorithms applied to portfolio selection or stock execution problems, such as \citet{wang2019continuous,wang2023reinforcement,huang2024mean}. The purpose of this strand of literature is fundamentally different from that of the present paper.}

Different motivations aside, the mathematical  framework for incorporating stochastic policies and entropy functions in continuous time is premised upon the same notion of \textit{relaxed controls} introduced by \citet{wang2018exploration}.   Numerous subsequent works have laid the mathematical foundations for randomized controls in continuous time. For example, \cite{bender2024continuous} and \cite{jia2025accuracy} confirm that state processes obtained by randomizing decisions at discrete time grids converge weakly to the so-called exploratory state process  in \citet{wang2018exploration} as the sampling frequency approaches infinity. \citet{tang2022exploratory} establish the well-posedness of the Hamilton-Jacobi-Bellman equation associated with the exploratory formulation in \citet{wang2018exploration} by a viscosity solution approach. The framework has also been extended to mean-field games \citep{guo2020entropy}, optimal stopping problems \citep{dianetti2025entropy,dai2024learning}, risk-sensitive problems \citep{jia2024continuous}, and time-inconsistent problems \citep{ddj2020calibrating}. All these papers adopt the additive form, directly adding the entropy function as a running reward.
%, and it is known that the optimal policy for a class of linear-quadratic problems is a Gaussian distribution.
Other forms of regularization or perturbation functions in continuous time have been studied by \citet{han2023choquet}, but they remain restricted to the additive form.

%Economically, RPU belongs to a special case of the general stochastic differential utility \citep{duffie1992stochastic} and extends APU to dynamic environments. In contrast to the well-known Epstein-Zin recursive utility \citep{epstein1989substitution}, which separates risk aversion and elasticity of intertemporal substitution, our RPU does not possess this feature. Instead, it incorporates the risk aversion coefficient into the weight on entropy regularization. Meanwhile, the entropy function---i.e., the logarithm of the density of the stochastic policy---already resembles the log utility for consumption. The connection between randomization-investment and consumption-investment problems is natural, as both involve trade-offs: randomization is analogous to consumption in that both lead to current gains and future losses. For consumption, the gain is immediate utility, and the loss is reduced capital for investment. For randomization, the gain is entropy, while the loss is increased volatility of future wealth due to added randomization, which lowers (risk-averse) utility. Intriguingly, introducing recursive entropy regularization gives rise to a history-dependent time preference on exploration, such that the more or longer the investor has explored in the past, the less weight she places on current exploration. This, in turn, results in an overall \emph{endogenous} scheme for weighting randomization, alongside an exogenously given temperature parameter that weighs the relative importance of entropy to the bequest utility.
%
%
%
%
%\subsection*{Structure of the Paper}
The remainder of this paper is organized as follows. Section \ref{setup} presents the market setup, formulates Merton's problem  and motivates the  recursive additive entropy utility. Section \ref{sec:theory analysis} proves the optimality of Gaussian randomized policy and discusses conditions under which the policy is biased or unbiased. The section further  provides an asymptotic analysis of the impact of the primary temperature parameter. Finally, Section \ref{sec:concl} concludes. Additional results and discussions are provided in the appendix.

\section{Problem Formulation}\label{setup}

\subsection{Market Environment and Investment Objective}

Consider a financial market with two available investment assets: a risk-free bond offering a constant interest rate $r$ and a risky stock (or market index). The stock price process $S_t$ %is assumed to be observable and
evolves according to the stochastic differential equation (SDE):
\begin{equation}
	\label{eq:model stock}
	\frac{\dd S_t}{S_t}=\mu(t,X_t)\dd t+\sigma(t,X_t)\dd B_t,\ S_0=s_0,
\end{equation}
where $B$ denotes a one-dimensional Brownian motion. The instantaneous return rate  $\mu_t\equiv \mu(t,X_t)$  and  volatility $\sigma_t\equiv \sigma(t,X_t)$ both depend on an observable stochastic market factor process $X_t$. The dynamics of $X_t$ are given by
\begin{equation}
	\label{eq:model factor}
	\dd X_t =m(t,X_t)\dd t+\nu(t,X_t)[\rho \dd B_t +\sqrt{1-\rho^2}\dd \tilde B_t],\ X_0 = x_0,
\end{equation}
where $\tilde B$ is an independent one-dimensional Brownian motion, and $\rho\in(-1,1)$ is a constant representing the correlation between the stock return and the market factor changes. As a result, the market is generally incomplete. We focus on the Markovian setting, where the functions $ \mu(\cdot,\cdot)$, $ \sigma(\cdot,\cdot)$, $ m(\cdot,\cdot)$, and $ \nu(\cdot,\cdot)$ are deterministic and continuous in both $t$ and $x$ such that the SDEs \eqref{eq:model stock}-\eqref{eq:model factor} admit a unique weak solution. This market environment is sufficiently general to encompass many widely studied incomplete market models as special cases---for example, the Gaussian mean return model and the stochastic volatility model examined  in \citet{wachter2002portfolio}, \citet{liu2007portfolio}, \citet{chacko2005dynamic,dai2018dynamic}, among others.\footnote{We assume there is only one stock and one factor for notational simplicity. There is no essential difficulty in extending to the multi-stock/factor case.}

The investor's actions are represented as a scalar-valued non-anticipative process $a=\{a_t\}_{t\in[0,T)}$, where $a_t$ denotes the fraction of total wealth allocated to the stock at time $t$. The associated self-financing wealth process $W^a$ evolves according to the SDE:
\begin{equation}\label{classical_wealth}
	\frac{\dd W_t^a}{W^a_t}=[r+(\mu(t,X_t)-r)a_t]\dd t+\sigma(t,X_t) a_t\dd B_t,\ \ W_0^a = w_0.
\end{equation}
It is important to note that the solvency constraint $W_t^a\geq 0$ almost surely for all $t\in [0,T]$, is automatically satisfied for any square integrable $a$.
The classical Merton investment problem seeks to maximize the expected bequest utility of terminal wealth $W_T^{a}$ by selecting an appropriate strategy $a$:
\begin{equation}
	\label{objective_functional_classical}
	\max_{a} \mathbb E\left[ U(W_T^{a})\right],
\end{equation}
subject to the dynamics \eqref{eq:model factor}-\eqref{classical_wealth}, where $U(\cdot)$ is the utility function.

In the main text, we focus on the constant relative risk aversion (CRRA) utility, given by $$U(w) = \frac{w^{1-\gamma} - 1}{1-\gamma},$$ where $1\neq\gamma > 0$ is  the relative risk aversion parameter. The additive constant  ``$-\frac{1}{1-\gamma}$'' does not affect the preference reflected; we include it here for two reasons. First, the form converges to the log utility as  $\gamma \rightarrow 1$. Second, including the constant will properly reflect the magnitude of the bequest utility when we combine it with the preference for randomization in Section \ref{subsec:recursive}. The case for the constant absolute risk aversion (CARA) utility function will be discussed in Appendix \ref{sect: cara}.

%The common approach to solve the problem \eqref{objective_functional_classical} is via dynamic programming and Hamilton-Jacobi-Bellman (HJB) equation, which requires defining a family of problems starting from different initial time and states. That is,
%\begin{equation}
%\label{eq:classical optimal value function def}
%V(t,x,w) = \max_{\{a_s\}_{s\in [t,T)}} \mathbb E\left[ e^{-\beta (T - t)} U(W_T^{a}) \Big| X_t = x, W_t^{a} = w \right] .
%\end{equation}
%However, the conventional approach relies on solving a non-linear partial differential equation (PDE) that depends on the knowledge about the market.

\subsection{Preference for Randomization}
\label{subsec:incentive}

As it stands, the classical Merton problem leads to a \emph{deterministic} optimal (feedback) policy. The agent has no appetite for randomization  because it would  increase the uncertainty in wealth, which is unfavored  by risk aversion. Therefore, to capture the preference for randomization, we need to add reward for randomization explicitly into the objective function. Moreover, with randomized choices, the resulting law of motion of the wealth process described by \eqref{classical_wealth} will also be fundamentally altered because $a_t$ will now be randomly sampled from a distribution that is independent of the market randomness ($B_t,\tilde B_t$). The resulting mathematical formulation is put forward and coined as the ``exploratory formulation" for general stochastic control problems by \citet{wang2018exploration} in the reinforcement learning (RL) context. We now adapt that formulation to the Merton problem.

The central idea of this formulation is to describe the portfolio rules using a probability distribution for randomization. Specifically, suppose the investor selects her action (portfolio) at time $t$ by sampling from a probability distribution $\bm\pi_t$, where
$\{\bm\pi_t\}_{t\in[0,T]}=:\bm\pi$ is a distribution-valued process referred to as a randomized {\it control}. Under such a control, the exploratory or randomized dynamics of the wealth process are given by
\begin{equation}
	\label{controlled_system}
	\frac{\dd W^{\bm\pi}_t}{W^{\bm\pi}_t} = \left[r+(\mu(t,X_t)-r)\operatorname{Mean}(\bm\pi_t)\right]\dd t +\sigma(t,X_t) \left[ \operatorname{Mean}(\bm\pi_t) \dd B_t+\sqrt{ \operatorname{Var}(\bm\pi_t) }\dd \bar{B}_t\right],\ W^{\bm\pi}_0 = w_0,
\end{equation}
where $\bar B$ is another Brownian motion, independent of both $B$ and $\tilde B$, representing the additional randomness introduced into the wealth process due to randomization. Equation \eqref{controlled_system} shows that randomization effectively raises the volatility of the resulting wealth process, and $\operatorname{Mean}(\bm\pi_t)$ plays a similar role as $a_t$ in the classical dynamics \eqref{classical_wealth}. Intuitively, $W^{\bm\pi}$ can be viewed as the ``average" of infinitely many wealth processes generated by portfolio processes repeatedly sampled from the {\it same} randomized control $\bm\pi$.
Another interpretation of $W^{\bm\pi}$ is the weak limit of wealth processes under piecewise constant portfolios where the portfolios are sampled from $\bm\pi$ only at discrete time points, as the mesh size of the sampling grid tends to zero; see \cite{bender2024continuous} and \cite{jia2025accuracy}.
The derivation of \eqref{controlled_system} is analogous to that in \cite{ddj2020calibrating} and \cite{daidata}; see a detailed explanation in Appendix A of \cite{daidata} and \citet{jia2025accuracy} for how \eqref{controlled_system} can be viewed as the limit of sampling randomized choice at a high frequency.

Second, to describe the preference for randomization, we adopt the entropy function to measure the level of randomness associated with a distribution $\bm\pi$, denoted as $
\mathcal{H}(\bm\pi)=-\int_{\mathbb{R}} \bm\pi(a)\log \bm\pi(a)\dd a
$ (known as the differential entropy).\footnote{Here we implicitly assume $\bm\pi$ to be absolutely continuous with respect to the Lebesgue measure so that the entropy can be properly defined.} %As a consequence, since randomization is always required, a flow of payoff $\mathcal{H}(\bm\pi_t)$ will be generated.  
For Merton's problem, the simplest way seems to be just adding running APU from randomization to the bequest utility, leading to the maximization of the following objective:
 \begin{equation}
 	\label{objective_functional_no_recursive}
 	\mathbb E\left[ \int_0^T \lambda   \mathcal{H}(\bm\pi_s) \dd s  + U(W_T^{\bm\pi}) \right],
 \end{equation}
where $\lambda>0$ is the \textit{temperature parameter} representing the  importance of randomization relative to the bequest utility, and $\{\lambda\mathcal{H}(\bm\pi_t)\}_{t\in[0,T]}$ is the flow of utility from randomization. The form \eqref{objective_functional_no_recursive} is a direct adaptation  of the APU arising from the static setting of \citet{fudenberg2015stochastic}; it  has also been chosen in the RL context by many algorithms as an incentive for exploration (e.g., \citealt{ziebart2008maximum,haarnoja2018soft} among many others).

As it turns out, however, it is difficult to analyze the randomized Merton problem under the objective \eqref{objective_functional_no_recursive}, where a closed-form solution is unavailable even for the Black-Scholes market with the CRRA utility (see Appendix \ref{appendix:what if constant temperature} for explanations). Even worse, Proposition \ref{prop:ill posed} below shows that the problem with the objective \eqref{objective_functional_no_recursive} is ill-posed if risk aversion is not sufficiently strong.
\begin{proposition}
\label{prop:ill posed}
If $\gamma\in (0,1)$, then the problem with the objective functional \eqref{objective_functional_no_recursive}, subject to the wealth dynamics \eqref{controlled_system}, is ill-posed with an infinite optimal value.
\end{proposition}
\proof{Proof.}
When $\gamma\in (0,1)$, $U(w) > -\frac{1}{1-\gamma}$ for any $w>0$. We consider a simple policy $\bm\pi_t = \mathcal N(0,v)$, where $v > 0$ is a constant. Then the objective function in \eqref{objective_functional_no_recursive} is greater than $\frac{\lambda T}{2}\log(2\pi e v)  - \frac{1}{1-\gamma} $. Letting $v\to\infty$ causes the objective value to diverge to infinity.
\qed
\endproof

The intuition of Proposition \ref{prop:ill posed} is that when the investor is not sufficiently risk-averse, there is not enough deterrence from over-randomization while  the utility from the entropy term can be unbounded. Indeed, the proof of the proposition indicates that the problem becomes ill-posed whenever the bequest utility $U$ is bounded from below. Hence, a na\"ive APU is not appropriate in the Merton setting  to yield reasonable economic predictions. %As a remark, from the proof of Proposition \ref{prop:ill posed}, we also notice that the constant part ``$-\frac{1}{1-\gamma}$'' in the bequest utility function does not affect this conclusion because the bequest utility is always bounded from below when $\gamma\in (0,1)$. In fact, based on the argument above, it is clear that the problem bocomes ill-posed whenever the utility function $U$ is bounded below. This result suggests that naively adding up the preference for randomization and the bequest utility does not lead to reasonable implications, and hence, it implies that there is a more sophisticated trade-off that should be included in the preference model.
%Perhaps the simplest way to modify the model is to introduce a wealth-dependent temperature parameter. In this setup, when wealth is low, the temperature correspondingly decreases, thereby preventing investments from becoming excessively randomized. Motivated by this, 
To remedy this problem, we introduce a different type of perturbed utility for randomization. 

\subsection{Recursive Perturbation Utility}
\label{subsec:recursive}
Let a function $\lambda(\cdot,\cdot) > 0$ be exogenously given, called a {\it primary} temperature function (it will be taken as a constant in Section \ref{subsec:asymptotic}).
Define by $J^{\bm\pi}=\{J^{\bm\pi}_t\}_{t\in [0,T]}$ the following {\it recursive} (entropy) perturbed utility (RPU) under a given randomized  control $\bm\pi$, which is an $\{\f_t\}_{t\in [0,T]}$-adapted process  satisfying
\begin{equation}
	\label{objective_functional}
	J^{\bm\pi}_t =\mathbb E\left[ \int_t^T \lambda(s,X_s)  [(1-\gamma)J_s^{\bm\pi} + 1] \mathcal H(\bm\pi_s) \dd s  +  U(W_T^{\bm\pi}) \Big| \f_t \right],
\end{equation}
where $W^{\bm\pi}$ is the wealth process under ${\bm\pi}$, determined by  \eqref{eq:model factor}  and \eqref{controlled_system}, and $\{\cal F\}_{t\in [0,T]}$ is the filtration generated by $\mathbf{B}: = (B,\bar{B},\tilde{B})^\top$.

The recursive form \eqref{objective_functional} effectively weighs the entropy utility endogenously: The term $\lambda(t,X_t) \left[ (1-\gamma)J_t^{\bm\pi} + 1 \right] $ can be viewed as a {\it utility-dependent} weight on randomization determined in a recursive way. %This form is motivated by the form of the bequest utility we specified. 
Notice that when $\gamma = 1$ (corresponding to the log utility function), the weight on the entropy term becomes independent of the utility, and \eqref{objective_functional} further reduces to \eqref{objective_functional_no_recursive} when $\lambda(t,x)\equiv \lambda$, i.e., the APU.

Under some proper assumptions (e.g., in \citealt{el1997backward}),  $J^{\bm\pi}$, which is an ${\cal F}_t$-adapted process,  solves the following backward stochastic differential equation (BSDE):
\begin{equation}
	\label{eq:bsde recursive general lambda}
	\dd J_t^{\bm\pi} = -\lambda_t \mathcal{H}(\bm\pi_t)  \left[(1-\gamma)J_t^{\bm\pi} + 1\right]   \dd t + \mathbf{Z}^{\bm\pi}_t \cdot \dd \mathbf{B}_t, \   J_T^{\bm\pi} = U({W_T^{\bm\pi}}),
\end{equation}
where (and henceforth) $\lambda_t\equiv \lambda(t,X_t)$ and  the notation $\cdot$ denotes the inner product. As explained in the Introduction, our formulation is motivated by the notion of ``recursive utility" in the economics literature \citep{epstein1989substitution,duffie1992stochastic}, with consumption being the counterpart of randomization.   In particular, the negative ``$\dd t$'' term in \eqref{eq:bsde recursive general lambda}, as a function of $\bm\pi_t$ and $J^{\bm\pi}_t$, can be written as $f(\bm\pi, J, \lambda) = \lambda\mathcal{H}(\bm\pi) - \left[ -\lambda (1-\gamma)\mathcal{H}(\bm\pi) \right] J$. This term  is the ``\textit{aggregator}'' in the recursive utility jargon, while the term $-\lambda (1-\gamma) \mathcal{H}(\bm\pi)$ corresponds to the  discount rate that depreciates the future utility $J$ into today's value.

We can formalize the above discussion to show a time preference on randomization implied by the definition of $J^{\bm\pi}$, in the same way as the  Uzawa utility  \citep{uzawa1968time} for consumption.  Indeed, we  explicitly solve \eqref{objective_functional} as
\begin{equation}\label{OF_expecation}
	J_t^{\bm\pi} = \E\left[\int_t^T e^{-\int_t^s -\lambda_{\tau} (1-\gamma)\mathcal{H}(\bm\pi_{\tau})  \dd \tau}\lambda_s \mathcal{H}(\bm\pi_s)\dd s + e^{-\int_t^T -\lambda_{\tau} (1-\gamma)\mathcal{H}(\bm\pi_{\tau})  \dd \tau} U(W_T^{\bm\pi}) \Big| \f_t \right].
\end{equation}
%where  $\lambda_t\equiv \lambda(t,X_t)$.
%	Note that if we replace the entropy $\mathcal H(\bm\pi_t)$ with the utility of consumption, then \eqref{OF_expecation} is the so-called Uzawa utility; see \cite{uzawa1968time}.  In economics, the Uzawa utility   captures the phenomenon that people's current time preference (discount factor)  may be affected by their past consumption, leading to an endogenous discount rate. More discussions on the economic and financial implications of the Uzawa utility can be found in \citet{bergman1985time}. In the RL context,
This expression implies  that the current weight on randomization  depends on  past randomization.
Specifically,
%The modified objective function \eqref{OF_expecation} turns out to be mathematically tractable and has financial meaning.
the (endogenous) temperature parameter for randomization now (compared to \eqref{objective_functional_no_recursive}) becomes  $\lambda_s e^{-\int_0^s -\lambda_{\tau} (1-\gamma)\mathcal{H}(\bm\pi_{\tau})  \dd \tau}$. In other words, at any time $s$, a discount  $e^{-\int_0^s -\lambda_{\tau} (1-\gamma)\mathcal{H}(\bm\pi_{\tau})  \dd \tau}$ is applied to  \eqref{objective_functional_no_recursive}. Moreover, empirical studies indicate that the typical risk-aversion parameter $\gamma > 1$ (see, e.g., \citealt{kydland1982time}), rendering $ -\lambda_{\tau} (1-\gamma)>0$. This implies that the more the investor has randomized  in the past, the less weight she places on current randomization, which is intuitive and sensible. Next, consider a small risk aversion $\gamma \in (0,1)$, where the discounting factor becomes negative, which seems  to incentivize larger randomization. However, there is a similar discount applied to the bequest utility. So, as  randomization increases, the constant part of the bequest utility ``$-\frac{1}{1-\gamma}$'', which now is a {\it negative} number, will go even more negative. This introduces a proper trade-off between randomization and bequest utility even when risk aversion is weak, avoiding the ill-posedness that occurs in the APU case.

Technically, the modified objective function \eqref{OF_expecation} makes our model mathematically tractable, as will be shown in the subsequent analysis.  %Second, in the terminal wealth term, we see that it includes a term

Henceforth denote $\mu_t\equiv \mu(t,X_t)$ and $\sigma_t\equiv \sigma(t,X_t)$. We are now ready to formulate our RPU Merton problem, by first formally introducing the set of admissible controls.
\begin{definition}
\label{def:admissible open loop}
An $\f_t$-adapted, probability-density-valued process 	$\bm\pi=\{\bm\pi_s, 0\leq s \leq T\}$ is called an (open-loop) admissible control, if
\begin{enumerate}[(i)]
\item for each $0\leq s\leq T$, $\bm\pi_s \in \mathcal P(\mathbb R)$ a.s., where $\mathcal P(\mathbb R)$ is the set of all probability densities on real numbers;
\item $\mathbb E\left[ \int_0^T|\sigma_s|^2(\operatorname{Mean}(\bm\pi_s)^2+\operatorname{Var}(\bm\pi_s))\dd s\right]+\mathbb E\left[ \int_0^T |\mu_s \operatorname{Mean}(\bm\pi_s)|\dd s\right]<\infty $;
\item  	$\mathbb E\left[e^{\int_0^T 2 \lambda_s |1-\gamma| |\mathcal H(\bm\pi_s)| ds}\right]+\mathbb E\left[ |U(W^{\bm\pi}_T)|^2\right]<\infty, $
where $\{X_s\}_{s \in [0,T]}$ and $\{W^{\bm\pi}_s\}_{s\in [0,T]}$ satisfy \eqref{eq:model factor} and \eqref{controlled_system}, respectively.
\end{enumerate}
\end{definition}

Given an admissible control $\bm\pi$ and an initial state pair $(w_0,x_0)$, we define the recursive utility $J_t^{\bm\pi}$ through \eqref{objective_functional},
%\begin{equation}
%\label{eq:reward}
%J_t^{\bm\pi}=\mathbb E\left[ \int_t^T \lambda(s,X_s) \left[ (1-\gamma)J_s^{\bm\pi}+ 1 \right]\mathcal H(\bm\pi_s) \dd s  +  U(W_T^{\bm\pi})\Big| \f_t \right],
%\end{equation}
where $\{X_t\}_{t\in [0,T]}$ and $\{W_t^{\bm\pi}\}_{t\in [0,T]}$ solve \eqref{eq:model factor} and \eqref{controlled_system}, respectively. 	
A technical question with the above definition is whether the entropy term in \eqref{objective_functional} has a positive weight, which is answered by the following proposition. Hence, our RPU indeed incentivizes randomization.
\begin{proposition}
\label{lemma:positive weight}
For any admissible control $\bm\pi$, we have $(1-\gamma)J^{\bm\pi}_t + 1 > 0$ almost surely.
\end{proposition}

\proof{Proof.}
%\begin{proof}
	Recall that RPU satisfies the  BSDE  \eqref{eq:bsde recursive general lambda}. Given an admissible control $\bm\pi$, the condition $(ii)$ in Definition \ref{def:admissible open loop} guarantees that the solution to \eqref{controlled_system} exists, and can be written as
	\[\begin{aligned}
		W_t^{\bm\pi}= w_0 \exp\bigg\{ \int_0^t &   \left[r+(\mu_s-r)\operatorname{Mean}(\bm\pi_s) + \frac{1}{2}\sigma_s^2\left(\operatorname{Mean}(\bm\pi_s)^2 + \operatorname{Var}(\bm\pi_s) \right) \right]\dd s \\
		& + \sigma_s \left[ \operatorname{Mean}(\bm\pi_s) \dd B_s+\sqrt{ \operatorname{Var}(\bm\pi_s) }\dd \bar B_s \right]\bigg\} >0 .
	\end{aligned} \]
	In addition, \eqref{eq:bsde recursive general lambda} is a linear BSDE, and the condition $(iii)$ in Definition \ref{def:admissible open loop} ensures that $U(W_T^{\bm\pi})$ and $\{\lambda(t,X_t)\mathcal{H}(\bm\pi_t)\}_{t\in [0,T]}$ are respectively square-integrable random variable and process. Hence by \citet[Proposition 2.2]{el1997backward}, \eqref{eq:bsde recursive general lambda} admits a unique square-integrable solution.
	%where $\mathbf{B}_t = (B_t,\bar{B}_t,\tilde{B}_t)^T$ is a three-dimensional standard Brownian motion and $\cdot$ means the inner product.
	
	Define $\tilde{Y}^{\bm\pi}_t: =  (1-\gamma)J_t^{\bm\pi} + 1 $.  Applying It\^o's lemma, we obtain that $\tilde{Y}^{\bm\pi}$ solves the following BSDE:
	\[ \dd \tilde{Y}^{\bm\pi}_t  = -(1-\gamma)\lambda(t,X_t) \mathcal{H}(\bm\pi_t)\tilde{Y}^{\bm\pi}_t  \dd t + (1-\gamma) \tilde{\mathbf{Z}}^{\bm\pi}_t \cdot \dd \mathbf{B}_t, \  \tilde{Y}^{\bm\pi}_T = ({W_T^{\bm\pi}})^{1-\gamma} > 0. \]
	From the comparison principle of linear BSDEs (see \citealt[Corollary 2.2]{el1997backward}), it follows that  $ \tilde{Y}^{\bm\pi}_t > 0 $ almost surely.%, by comparing to the solution to a trivial BSDE: $\dd Y_t = 0 \dd t + 0  \cdot \dd B_t, Y_T = 0$.}
	\qed
\endproof
%\end{proof}

\medskip

To apply dynamic programming  to problem \eqref{objective_functional}, we further restrict our attention to feedback policies. A (randomized) {\it feedback policy} (or simply a {\it policy}) $\bm\pi=\bm\pi(\cdot,\cdot,\cdot)$ is a density-valued function of time and state,
under which \eqref{eq:model factor} and \eqref{controlled_system} become a Markovian system. For any initial time $t$ and initial state $(w,x)$, a policy $\bm\pi$ induces the open-loop control  $\bm\pi_s = \bm\pi(s,W_s^{\bm\pi},X_s)$, where $\{X_s\}_{s\in [t,T]}$ and $\{W_s^{\bm\pi}\}_{s\in [t,T]}$ are the solutions to the corresponding Markovian system given $W_t^{\bm\pi}=w$ and $X_t=x$. Denote by $\Pi$ the set of policies that induce admissible open-loop controls.

Given $\bm\pi\in \Pi$, define its {\it value function} $V^{\bm\pi}(\cdot,\cdot,\cdot)$ as
\begin{equation}\label{OF_expecation_txw}
\begin{aligned}
V^{\bm\pi}(t,w,x): = & \E\big[\int_t^T e^{-\int_t^s -\lambda (1-\gamma)\mathcal{H}(\bm\pi_{\tau})  \dd \tau}\lambda \mathcal{H}(\bm\pi_s)\dd s + e^{-\int_t^T -\lambda (1-\gamma)\mathcal{H}(\bm\pi_{\tau})  \dd \tau} U(W_T^{\bm\pi}) \Big|W_t^{\bm\pi} = w, X_t = x \big],\\
&\;\;\;\;(t,w,x)\in [0,T]\times \R_+\times \R.
\end{aligned}
\end{equation}
The Feynman-Kac formula yields that this function satisfies the PDE:
\begin{equation}
\label{eq:pde for j pi}
\begin{aligned}
\frac{\partial V^{\bm\pi}(t,w,x)}{\partial t} & +  \Big( r + \big(\mu(t,x) - r\big)\operatorname{Mean}\left( \bm\pi(t,w,x) \right) \Big)w  V^{\bm\pi}_w(t,w,x) \\
& + \frac{1}{2}\sigma^2(t,x)\Big( \operatorname{Mean}\left( \bm\pi(t,w,x) \right)^2 +  \operatorname{Var}\left( \bm\pi(t,w,x) \right) \Big)w^2  V^{\bm\pi}_{ww}(t,w,x) \\
& + m(t,x) V^{\bm\pi}_x(t,w,x) + \frac{1}{2}\nu^2(t,x) V^{\bm\pi}_{xx}(t,w,x)+ \rho\nu(t,x)\sigma(t,x)\operatorname{Mean}(\bm\pi) w  V^{\bm\pi}_{xw}(t,w,x) \\
& + \lambda(t,x) \mathcal{H}\left(\bm\pi(t,w,x)\right)\big[ (1-\gamma)V^{\bm\pi}(t,w,x) + 1\big] = 0,\ V^{\bm\pi}(T,w,x) = U(w).
\end{aligned}
\end{equation}

Using the relation between BSDEs and PDEs, we can represent the recursive utility $J^{\bm\pi}$ via the value  function $V^{\bm\pi}$:
\begin{equation}\label{connection}
J_t^{\bm\pi}=V^{\bm\pi}(t,W_t^{\bm\pi},X_t), \;\;\mbox{a.s.},\;\;\;t\in [0,T],
\end{equation}
where $\{X_t\}_{t \in [0,T]}$ and $\{W^{\bm\pi}_t\}_{t\in [0,T]}$ satisfy \eqref{eq:model factor} and \eqref{controlled_system} with $W_0^{\bm\pi}=w_0$ and $X_0=x_0$, respectively.	

Finally, we define the {\it optimal} value function  as
\begin{equation}\label{optimalvalue}
V(t,w,x):=\sup_{\bm\pi\in\Pi} J^{\bm\pi}(t,w,x),\;\;(t,w,x)\in [0,T]\times \R_+\times \R.
\end{equation}

\section{Theoretical Analysis}
\label{sec:theory analysis}
\subsection{Gaussian Randomization}

It is straightforward, as in \cite{wang2018exploration}, to derive that the optimal value function $V$ satisfies the following HJB equation via dynamic programming:
\begin{equation}
\label{eq:hjb 0}
\begin{aligned}
& \frac{\partial V}{\partial t} + \sup_{\bm\pi\in \mathcal P(\mathbb R)}\Bigg\{ \Big( r + \big(\mu(t,x) - r\big)\operatorname{Mean}(\bm\pi) \Big)wV_w + \frac{1}{2}\sigma^2(t,x)\Big( \operatorname{Mean}(\bm\pi)^2 +  \operatorname{Var}(\bm\pi) \Big)w^2V_{ww} \\
& + m(t,x) V_x + \frac{1}{2}\nu^2(t,x)V_{xx} + \rho\nu(t,x)\sigma(t,x)\operatorname{Mean}(\bm\pi) wV_{wx} + \lambda(t,x) \mathcal{H}(\bm\pi)\big[ (1-\gamma)V + 1\big] \Bigg\} = 0,
\end{aligned}
\end{equation}
with the terminal condition $V(T,w,x) =  U(w) = \frac{w^{1-\gamma} - 1}{1-\gamma}$.

At first glance, equation \eqref{eq:hjb 0}  is a highly nonlinear PDE and appears  hard to analyze. However, we can reduce it to a simpler PDE based on which the optimal randomized  policy can be explicitly represented.
\begin{theorem}
\label{thm:hjb solution}
Suppose $u$ is a classical solution of the following PDE
\begin{equation}
\label{eq:hjb 4}
\begin{aligned}
	\frac{\partial u}{\partial t} & + (1-\gamma)r  + m(t,x)u_x + \frac{1}{2}\nu^2(t,x)(u_{xx} + u_x^2) + \frac{(1-\gamma)\lambda(t,x)}{2}\log \frac{2\pi \lambda(t,x)}{\gamma  \sigma^2(t,x)}  \\
	& + \frac{1-\gamma}{2\gamma}\left[ \frac{\big( \mu(t,x) - r \big)^2}{\sigma^2(t,x)} + \frac{2\rho\big( \mu(t,x)-r \big)\nu(t,x)}{\sigma(t,x)}u_x  + \rho^2\nu^2(t,x)u_x^2\right]  = 0,
\end{aligned}
\end{equation}
with the terminal condition $u(T,x) = 0$, and the derivatives of $u$ up to the second order have polynomial growth. Let $\bm\pi^*(t,x)$ be a normal distribution with
\begin{equation}
\label{eq:optimal exploratory policy via u}
\operatorname{Mean}(\bm\pi^*(t,x)) = \frac{\mu(t,x) - r}{\gamma \sigma^2(t,x)} +\frac{\rho\nu(t,x)}{\gamma\sigma(t,x)}u_x(t,x),\ \operatorname{Var}(\bm\pi^*(t,x)) = \frac{\lambda(t,x)}{\gamma\sigma^2(t,x)}.
\end{equation}
If $\bm\pi^*\in\Pi$, then it is the optimal randomized  policy. Furthermore, the optimal value function is  $V(t,w,x) = \frac{w^{1-\gamma}e^{u(t,x)} - 1}{1-\gamma}$.
\end{theorem}

\proof{Proof.}
First, let us show how to deduce \eqref{eq:optimal exploratory policy via u} from \eqref{eq:hjb 0}.
Note that the  ``supermum'' in \eqref{eq:hjb 0} can be achieved via a two-stage optimization procedure: first maximize over distributions with a {\it fixed}  mean and variance pair, and then maximize over all possible such pairs. For the first problem, the entropy is maximized at a normal distribution with the fixed mean and variance, with the maximum entropy value $\mathcal{H}(\bm\pi) = \frac{1}{2}\log\big( 2 \pi e \operatorname{Var}(\bm\pi) \big)$ where $\operatorname{Var}(\bm\pi)$ is the given fixed variance.  Therefore, \eqref{eq:hjb 0} can be simplified as
\begin{equation}
	\label{eq:hjb 1}
	\begin{aligned}
		& \frac{\partial V}{\partial t} + \sup_{\operatorname{Mean}(\bm\pi),\operatorname{Var}(\bm\pi)}\Bigg\{ \Big( r + \big(\mu(t,x) - r\big)\operatorname{Mean}(\bm\pi) \Big)wV_w + \frac{1}{2}\sigma^2(t,x)\Big( \operatorname{Mean}(\bm\pi)^2 +  \operatorname{Var}(\bm\pi) \Big)w^2V_{ww} \\
		& + m(t,x) V_x + \frac{1}{2}\nu^2(t,x)V_{xx} + \rho\nu(t,x)\sigma(t,x)\operatorname{Mean}(\bm\pi)w V_{wx} + \frac{\lambda(t,x)}{2}\log\big( 2\pi e \operatorname{Var}(\bm\pi) \big) \big[  (1-\gamma)V + 1\big]  \Bigg\}
		 = 0.
	\end{aligned}
\end{equation}

To analyze the above PDE, we start with  the ansatz that $V(t,w,x) = \frac{w^{1-\gamma}}{1-\gamma}v(t,x) - \frac{1}{1-\gamma}$ for some  function $v$. Then \eqref{eq:hjb 1} becomes
\begin{equation}
	\label{eq:hjb 2}
	\begin{aligned}
		& \frac{w^{1-\gamma}}{1-\gamma} \frac{\partial v}{\partial t} +  \sup_{\operatorname{Mean}(\bm\pi),\operatorname{Var}(\bm\pi)}\Bigg\{ \Big( r + \big(\mu(t,x) - r\big)\operatorname{Mean}(\bm\pi) \Big)w^{1-\gamma}v - \frac{\gamma}{2}\sigma^2(t,x)\Big( \operatorname{Mean}(\bm\pi)^2 +  \operatorname{Var}(\bm\pi) \Big)w^{1-\gamma}v \\
		& + m(t,x) \frac{w^{1-\gamma}}{1-\gamma}v_x + \frac{1}{2}\nu^2(t,x)\frac{w^{1-\gamma}}{1-\gamma}v_{xx}  + \rho\nu(t,x)\sigma(t,x)\operatorname{Mean}(\bm\pi) w^{1-\gamma}v_x  + \frac{\lambda(t,x)}{2}\log\big( 2 \pi e \operatorname{Var}(\bm\pi) \big) w^{1-\gamma}v \Bigg\}  =0.
	\end{aligned}
\end{equation}

If $ v > 0$ (to be verified later), then the first-order conditions of the maximization problem on the left-hand side of the above equation yield the maximizers
\begin{equation}\label{mv}
	\operatorname{Mean}(\bm\pi^*) = \frac{\mu(t,x) - r}{\gamma \sigma^2(t,x)} +\frac{\rho\nu(t,x)v_x(t,x)}{\gamma\sigma(t,x)v(t,x)},\;\;\;
	\operatorname{Var}(\bm\pi^*) = \frac{\lambda(t,x)}{\gamma \sigma^2(t,x)}.
\end{equation}

%Here, we make  the assumption that  $v= (1-\gamma)g + \kappa$, which will be verified later. Hence the optimal exploratory variance is $\operatorname{Var}(\bm\pi^*) = \frac{\lambda(t,x)}{\gamma \sigma^2(t,x)}$.

Plugging the above into \eqref{eq:hjb 2} and
canceling the term $w^{1-\gamma}$, we have that $v$ satisfies the following nonlinear PDE:
\begin{equation}
	\label{eq:hjb 3}
	\begin{aligned}
		\frac{\partial v}{\partial t} & +  (1-\gamma)r  v + m(t,x)v_x + \frac{1}{2}\nu^2(t,x)v_{xx} + \frac{(1-\gamma)\lambda(t,x) v}{2}\log \frac{2 \pi \lambda(t,x)}{\gamma\sigma^2(t,x)}  \\
		& + \frac{1-\gamma}{2\gamma}\Big[ \frac{\big( \mu(t,x) - r \big)^2}{\sigma^2(t,x)} v + \frac{2\rho\big( \mu(t,x)-r \big)\nu(t,x)}{\sigma(t,x)}v_x  + \rho^2\nu^2(t,x)\frac{v_x^2}{v}\Big]  = 0 ,
	\end{aligned}
\end{equation}
with the terminal condition $v(T,x) = 1$.

%Now consider the case when $\gamma = 1$.  Recall that $\frac{w^{1-\gamma}-1}{1-\gamma} \to \log w$. Consider all $\log w$ terms in \eqref{eq:hjb 2},\footnote{{\bf The argument here is problematic. \eqref{eq:hjb 2} is not well defined when $\gamma = 1$. Suggest removing this case throughout the paper.}} we can easily obtain that $v(t,x)$ satisfies that
%\[ v_t + m(t,x)v_x + \frac{1}{2}\nu^2(t,x)v_{xx} = \beta v, \]
%which is consistent with \eqref{eq:hjb 3}. And we can see the solution is $v(t,x) = e^{-\beta(T-t)}. $In this case, $g(t,x)$ satisfies
%\begin{equation}
%\label{eq:eq for g log}
%\begin{aligned}
%& g_t + m(t,x)g_x + \frac{1}{2}\nu^2(t,x)g_{xx}+  \frac{\lambda(t,x)1}{2}\log\big( 2\bm\pi \frac{\lambda(t,x)1}{\gamma\sigma^2(t,x)v} \big) \\
%& + rv + \frac{\gamma\sigma^2(t,x)v}{2}\big( \frac{\mu(t,x)-r}{\gamma\sigma^2(t,x)} + \frac{\rho\nu(t,x)v_x}{\gamma\sigma(t,x)v} \big)^2 = \beta g
%\end{aligned}
%\end{equation}
%and the associated optimal variance is $\frac{\lambda(t,x)\kappa(t,x)}{\gamma\sigma^2(t,x)v}$. Recall that we choose $\kappa(t,x) = e^{-\beta(T-t)}$, it turns out that $v=\kappa$ which is consistent with our ansatz.

Letting  $u$ be the solution to \eqref{eq:hjb 4}, we can directly verify that $v(t,x) = e^{u(t,x)} >0$ solves \eqref{eq:hjb 3}. The desired expression \eqref{eq:optimal exploratory policy via u} now follows from \eqref{mv}.

Next we prove that the policy with \eqref{eq:optimal exploratory policy via u} is optimal and $V(t,w,x) = \frac{w^{1-\gamma}e^{u(t,x)} - 1}{1-\gamma}$ is the optimal value function.
%Consider the following BSDE:
%\begin{equation}
%	\label{eq:general lambda optimal bsde}
%	\left\{  \begin{aligned}
	%		\dd Y_t^{*} = & -\Bigg\{ (1-\gamma)\Big[ \frac{\lambda(t,X_t)}{2}\log\big(\frac{2 \pi \lambda(t,X_t)}{\gamma\sigma_t^2} \big) + r + \frac{(\mu_t - r + \rho\sigma_t Z^{*}_t)^2}{2\gamma\sigma_t^2} \Big]  + \frac{1}{2}{Z_t^{*}}^2\Bigg\} \dd t + Z_t^{*} \dd B_t^X\\
	%		Y^{*}_T = & 0.
	%	\end{aligned}\right. ,
%\end{equation}
%where $\dd B^X_t = \rho\dd B_t + \sqrt{1-\rho^2}\dd \tilde{B}_t$ is the Brownian motion increment associated with the market factor $X$.
%
%Comparing \eqref{eq:general lambda optimal bsde} and \eqref{eq:hjb 4}, we observe that $Y^{*}_t = u(t, X_t)$ solves BSDE \eqref{eq:general lambda optimal bsde}.
%
%
%
%Comparing two BSDEs \eqref{eq:general lambda optimal bsde} and \eqref{eq:bsde value function under pi}.
Denote $K_t^{\bm\pi} = \int_0^t \lambda(s,X_s) (1-\gamma)\mathcal{H}(\bm\pi_s)\dd s$. Apply It\^o's formula to $e^{K_t} V(t,W^{\bm\pi}_t,X_t)$ to get
\begin{equation*}
	\begin{split}
		& \dd\left( e^{K_t^{\bm\pi}} V(t,W^{\bm\pi}_t,X_t) \right) \\
		=&\bigg\{ \frac{\partial V}{\partial t} +\left[ r + (\mu(t,X_t)-r)\operatorname{Mean}(\bm\pi_t)\right] W^{\bm\pi}_t V_w+m(t,X_t) V_x +\frac{1}{2}\sigma^2(t,X_t)\left[ \operatorname{Mean}(\bm\pi_t)^2+\operatorname{Var}(\bm\pi_t)\right] (W^{\bm\pi}_t)^2 V_{ww}\\
		& +\frac{1}{2}\nu^2(t,X_t) V_{xx} + \rho\nu(t,X_t)\sigma(t,X_t)\operatorname{Mean}(\bm\pi_t)W^{\bm\pi}_t  V_w + \lambda(1-\gamma)\mathcal{H}(\bm\pi_t) V \bigg\} e^{K_t^{\bm\pi}} \dd t\\
		&+ \bigg\{ \left[ \sigma(t,X_t)\operatorname{Mean}(\bm\pi_t)W^{\bm\pi}_t V_w+\rho \nu(t,X_t) V_x \right] \dd B_t \\
		&+  \sqrt{1-\rho^2}\nu(t,X_t) V_x \dd\tilde B_t+\sqrt{\operatorname{Var}(\bm\pi_t)} \sigma(t,X_t)W^{\bm\pi}_t V_w \dd \bar B_t\bigg\}e^{K_t^{\bm\pi}} .
	\end{split}
\end{equation*}

Define a sequence of stopping times $\tau_n = \inf\{ t\geq 0: |X_t| \vee (W_t^{\bm\pi})^{1-\gamma} \vee e^{K_t^{\bm\pi}} \geq n  \}$.
Then by \eqref{eq:hjb 1}, we have
\[ \begin{aligned}
	& e^{K_{T\wedge \tau_n}^{\bm\pi}} V(T\wedge \tau_n,W^{\bm\pi}_{T\wedge \tau_n},X_{T\wedge \tau_n}) -  V(0,W^{\bm\pi}_0,X_0) \\
	\leq & \int_0^{T\wedge \tau_n} -\lambda(s,X_s) \mathcal{H}(\bm\pi_s)   e^{K_s^{\bm\pi}} \dd s  + e^{K_s^{\bm\pi}} (W^{\bm\pi}_s)^{1-\gamma} e^{u(s,X_s)}\bigg\{ \left[ \sigma(s,X_s)\operatorname{Mean}(\bm\pi_s) +\rho \nu(s,X_s) \frac{u_x(s,X_s)}{1-\gamma}  \right] \dd B_s \\
	&+  \sqrt{1-\rho^2}\nu(s,X_s) \frac{u_x(s,X_s)}{1-\gamma}  \dd\tilde B_s+\sqrt{\operatorname{Var}(\bm\pi_s)} \sigma(s,X_s) \dd \bar B_s\bigg\} ,
\end{aligned} \]
while ``$=$'' holds when $\bm\pi = \bm\pi^*$ as given in \eqref{eq:optimal exploratory policy via u}.

Recall that $u$ and its derivatives have polynomial growth in $x$; hence we have an estimate about the quadratic variation term above
\[ \begin{aligned}
	& \E\Bigg[\int_0^{T\wedge \tau_n} e^{2K_s^{\bm\pi}} (W^{\bm\pi}_s)^{2(1-\gamma)} e^{2u(s,X_s)}\bigg\{\left[ \sigma(s,X_s)\operatorname{Mean}(\bm\pi_s) +\rho \nu(s,X_s) \frac{u_x(s,X_s)}{1-\gamma}  \right]^2 \\
	&  + (1-\rho^2)\nu^2(s,X_s) \frac{u_x^2(s,X_s)}{(1-\gamma)^2} +  \operatorname{Var}(\bm\pi_s)  \sigma^2(s,X_s) \bigg\} \dd s\Bigg] \\
	\leq & \E\left[\int_0^{T} C_n \bigg\{ \sigma^2(s,X_s) \left[ \operatorname{Mean}(\bm\pi_s)^2 +  \operatorname{Var}(\bm\pi_s) \right]  + 1 \bigg\} \dd s\right]
\end{aligned}  \]
for some constants $C_n>0$.
Therefore, the expectations of the corresponding stochastic integrals are all 0. When $\bm\pi$ is admissible, by the condition ($ii$) in Definition \ref{def:admissible open loop}, we deduce
\[\begin{aligned}
	& V(0, w_0, x_0) \\
	\geq & \E\left[ e^{K_{T\wedge \tau_n}^{\bm\pi}} V(T\wedge \tau_n,W^{\bm\pi}_{T\wedge \tau_n},X_{T\wedge \tau_n})   +\int_0^{T\wedge \tau_n} \lambda(s,X_s) \mathcal{H}(\bm\pi_s)   e^{K_s^{\bm\pi}} \dd s\right] \\
	= & \E\left[ e^{K_{T\wedge \tau_n}^{\bm\pi}} \frac{(W^{\bm\pi}_{T\wedge \tau_n})^{1-\gamma}}{1-\gamma} e^{u(T\wedge \tau_n, X_{T\wedge \tau_n})} - \frac{  e^{K_{T\wedge \tau_n}^{\bm\pi}} }{1-\gamma}   +\int_0^{T\wedge \tau_n} \lambda(s,X_s) \mathcal{H}(\bm\pi_s)   e^{K_s^{\bm\pi}} \dd s\right] \\
	= &  \E\left[ e^{K_{T}^{\bm\pi}} \frac{(W^{\bm\pi}_{T})^{1-\gamma}}{1-\gamma} \one_{\{ \tau_n > T \}} + e^{K_{\tau_n}^{\bm\pi}} \frac{(W^{\bm\pi}_{ \tau_n})^{1-\gamma}}{1-\gamma} e^{u(\tau_n, X_{\tau_n})} \one_{\{ \tau_n \leq T \}} - \frac{  e^{K_{T\wedge \tau_n}^{\bm\pi}} }{1-\gamma}   +\int_0^{T\wedge \tau_n} \lambda(s,X_s) \mathcal{H}(\bm\pi_s)   e^{K_s^{\bm\pi}} \dd s\right].
\end{aligned}  \]
By the condition ($iii$) in Definition \ref{def:admissible open loop}, we have
\[\begin{aligned}
	& 	\E\left[\int_0^T \lambda(t,X_t) |\mathcal{H}(\bm\pi_t)|   e^{K_t^{\bm\pi}} \dd t \right] \leq  \E\left[\int_0^T \lambda(t,X_t) |\mathcal{H}(\bm\pi_t)|   e^{\int_0^t \lambda |1-\gamma| |\mathcal{H}(\bm\pi_s) |\dd s} \dd t \right]\\
	\leq & \E\left[ \int_0^T \lambda(t,X_t) |\mathcal{H}(\bm\pi_t)|\dd t  e^{\int_0^T \lambda |1-\gamma| |\mathcal{H}(\bm\pi_s) |\dd s} \right] \\
	\leq & \left( \E\left[ \left( \int_0^T \lambda(t,X_t) |\mathcal{H}(\bm\pi_t)| \dd t \right)^2 \right]  \right)^{1/2} \left( \E\left[ e^{\int_0^T 2\lambda |1-\gamma| |\mathcal{H}(\bm\pi_s) |\dd s} \right] \right)^{1/2} < \infty,
\end{aligned}   \]
and
\[ \E\left[ e^{K^{\bm\pi}_T} (W^{\bm\pi}_{T})^{1-\gamma} \right] \leq \left( \E\left[ e^{\int_0^T 2\lambda |1-\gamma| |\mathcal{H}(\bm\pi_s) |\dd s} \right] \right)^{1/2} \left( \E\left[  (W^{\bm\pi}_{T})^{2-2\gamma} \right] \right)^{1/2} < \infty .\]

Taking $\limsup_{n\to \infty}$, we conclude from the dominated convergence theorem and Fatou's lemma that
\[ V(0,w_0,x_0) \geq \E\left[ e^{K_T^{\bm\pi}} U(W_T^{\bm\pi}) + \int_0^T \lambda(s,X_s)\mathcal{H}(\bm\pi_s) e^{K_s^{\bm\pi}}\dd s  \right] = J_0^{\bm\pi} . \]

To show ``$=$'' for $\bm\pi^*$, we can directly verify that $V_t =V(t,W_t^{\bm\pi^*},X_t)$ satisfies the BSDE \eqref{eq:bsde recursive general lambda}, and hence $V_t$ is the recursive utility process under $\bm\pi^*$. This completes our proof.
\qed
\endproof

\bigskip

Some remarks are in order. First,
the optimal randomized policy follows {\it Gaussian} with its mean and variance explicitly given via the PDE \eqref{eq:hjb 4}, even if the current setup is neither LQ nor mean-variance. The policy variance depends only on the exogenously given primary temperature function $\lambda$, the risk aversion parameter $\gamma$, and the instantaneous variance function $\sigma^2$. A larger primary temperature  increases the  level of randomization, while a greater risk aversion or volatility reduces it. These results are mathematically consistent with those in \citet{wang2018exploration} and \cite{ddj2020calibrating}, even though in a different context (i.e., RL), who consider an LQ control and an equilibrium mean-variance criterion, respectively. Second, the mean of the optimal policy consists of two parts: a myopic part $\frac{\mu(t,x)-r}{\gamma\sigma^2(t,x)}$ independent of randomization, and a hedging part represented by
\[ \frac{\rho\nu(t,x)}{\gamma\sigma(t,x)}u_x(t,x) = \frac{\operatorname{Cov}(\dd X_t,\dd \log S_t )}{\operatorname{Var}(\dd \log S_t)}u_x(t,X_t). \]
Note that hedging is needed due to the presence of the factor $X$  even in the classical Merton setting; yet the level of hedging is {\it altered} by the agent randomization because $u$ depends on the choice of $\lambda$ via the PDE \eqref{eq:hjb 4}.
As a result, unlike the previous works (e.g., \citealt{wang2018exploration}, \citealt{wang2019continuous}, and \citealt{ddj2020calibrating}) where the optimal policies depend on randomization only through variance and are thus unbiased, the optimal policies here are generally {\it biased} and the degree of biasedness depends on that of randomization. We will investigate this feature in more detail in the following subsections. Finally,
the resulting weight of randomization in the objective function is $\lambda(t,x)[(1-\gamma)V(t,w,x) + 1]=\lambda(t,x)w^{1-\gamma}e^{u(t,x)} > 0$. When $\lambda(t,x)$ is a constant and $\gamma=1$, $u \equiv 0$ is the solution to \eqref{eq:hjb 4} and, as a consequence, the weight reduces to a constant as in the APU objective  \eqref{objective_functional_no_recursive}.

%Finally, from the optimal condition implied in \eqref{eq:hjb 4}, one sees that  the optimal variance is
%$$
%\operatorname{Var}({\bm\pi}^*)=\frac{\lambda(t,x)(1-\gamma )V+1)}{-\sigma^2w^2V_{ww}}.
%$$
%The denominator captures the sensitivity of the utility to changes in the variance of the randomization. Thus, the equation indicates that investors determine the optimal level of randomization by balancing two factors: one is the reward implied by entropy, and the other is the utility loss caused by randomization.  In general, the denominator depends on the wealth level. However, since the control variable is the investment ratio, it is more reasonable for it to remain stable across different wealth levels. This justifies the use of a utility- or wealth-adjusted temperature parameter.

\subsection{When is the Optimal Randomized  Policy Unbiased?}\label{sec:unbi}

When $\lambda \equiv 0$, the optimal Gaussian distribution in Theorem \ref{thm:hjb solution} degenerates into the Dirac measure concentrating on the mean, $\frac{\mu(t,x) - r}{\gamma \sigma^2(t,x)} +\frac{\rho\nu(t,x)}{\gamma\sigma(t,x)}u_x(t,x)$, where $u$ solves \eqref{eq:hjb 4}  with $\lambda \equiv 0$. This is Merton's strategy for the classical problem (in the incomplete market) without preference for randomization. We refer to the case $\lambda\equiv 0$ as the ``classical case" in the rest of this paper. As we have pointed out, unlike the existing results, the mean of the optimal randomized policy does not generally coincide with that of the classical counterpart due to an interaction between randomization and hedging.\footnote{One should, however, note that  biased exploratory/randomized policies are common in the RL literature for various reasons. For example, the $\epsilon$-greedy policy  is a convex combination of the (classically) optimal one and a purely random policy; hence it is biased.}

There are, however,  special circumstances even in our setting where the optimal Gaussian policy becomes unbiased.
According to \eqref{eq:optimal exploratory policy via u}, the part that causes biases is the hedging demand, $\frac{\rho\nu(t,x)}{\gamma\sigma(t,x)}u_x(t,x)$.
Hence, if $\nu \equiv  0$ or $\rho=0$ (i.e., the  factor $X$ is deterministic or evolves independently from the stock price), then this part vanishes and the optimal  policy becomes unbiased.  In these cases, changes in the market factor do not affect the stock return or there is no hedging need against  the factor,  and thus a myopic policy irrelevant to our choice of $\lambda$ is dynamically optimal.
Next, note that the only difference between the classical and RPU problems is reflected by the extra term $\frac{(1-\gamma)\lambda(t,x)}{2}\log\frac{2\pi\lambda(t,x)}{\gamma\sigma^2(t,x)}$ in the PDE \eqref{eq:hjb 4}. If $\gamma=1$ (log-utility), then, for any choice of the function $\lambda$, this extra term vanishes, and hence the unbiasedness holds.
More generally, if one chooses $\lambda$ such that
$\lambda(t,x)\log\frac{2\pi\lambda(t,x)}{\gamma\sigma^2(t,x)}$ is  independent of $x$, then,  taking derivative in $x$ on both sides of \eqref{eq:hjb 4} yields that  $u_x$ satisfies the same PDE regardless of whether $\lambda=0$ or not. This in turn implies that the hedging part $\frac{\rho\nu(t,x)}{\gamma\sigma(t,x)}u_x(t,x)$ is independent of $\lambda$, leading to the unbiasedness of the optimal Gaussian policy.

%A question remains from the above analysis: How to make the term $\lambda(t,x)\log\frac{2\pi\lambda(t,x)}{\gamma\sigma^2(t,x)}$ independent of $x$ without knowing the functional form of $\sigma$. When $\sigma$ is unknown yet is known to be independent of $x$ (so the instantaneous volatility rate is deterministic), we can choose $\lambda(t,x)$ to be independent of $x$ so that the above term is independent of $x$. %More generally,  recall that $\sigma^2$ stands for the instantaneous variance of stock return, which  we may use additional market data to approximate (e.g., VIX for S\&P 500 or implied volatility from option prices). Therefore, if such a function approximation is possible, then we can choose the function $\lambda$ accordingly to make the above term independent of $x$, resulting in unbiasedness.

A discussion of the unbiasedness will also be given from the BSDE perspective in Appendix \ref{app:bsde}.

\subsection{An Asymptotic Analysis on $\lambda$}
\label{subsec:asymptotic}
%As discussed earlier, taking $\lambda(t,x) = \lambda$ is a parsimonious choice that may be easy to work with additional market data on volatility. Therefore, from now on, we fix it to be a constant.
%We use randomized actions to train and improve policies. {\it After} the learning procedure, one implements  the learned policies. In this step, a randomized policy is no longer necessary nor suitable; instead, the agent can use   its mean (which is a deterministic policy) for actual execution  and consider the corresponding expected utility of  terminal wealth without entropy. However, because  the mean of the exploratory optimal policy is generally a biased estimate of the optimal policy of the classical model,

The preference for randomization induces a different objective function for the investor and in general makes the optimal policy biased. In particular, randomization increases the uncertainty of the wealth process and thus lowers the bequest utility. It is interesting to investigate the financial losses due to this preference for randomization and quantify the impact of such a bias. We carry out this investigation by an asymptotic analysis on the PDE \eqref{eq:hjb 4} in the small parameter $\lambda$, which is henceforth assumed to be a {\it constant} (instead of a function) $\lambda(t,x) \equiv \lambda$, leading to asymptotic expansions of the optimal policy along with its value function.

We denote by $u^{(0)}$ the solution to \eqref{eq:hjb 4} with $\lambda=0$ and by $V^{(0)}$ the optimal value function for the classical problem (i.e., when $\lambda=0$).
It follows from Theorem \ref{thm:hjb solution} that  $V^{(0)}(t,w,x) = \frac{w^{1-\gamma}e^{u^{(0)}(t,x)} - 1}{1-\gamma}$. For any $\lambda>0$, let $\bm\pi^*$ be the  optimal randomized policy, and $V^{(\lambda)}$ be the value function of the original {\it non-randomized} problem under the deterministic policy that is the mean of $\bm\pi^*$:
\begin{equation}
\label{eq:payoff under learned policy}
V^{(\lambda)}(t,w,x) = \E\left[ U\left(W_T^{\widehat{\bm\pi^*}}\right) \Big| W_t^{\widehat{\bm\pi^*}} = w, X_t = x \right],
\end{equation}
where $\widehat{\bm\pi^*}(t,x) = \mathcal{N}\left( \operatorname{Mean}\left(\bm\pi^*(t,x)\right), 0\right)$.

The results of our asymptotic analysis involve several functions, among which
$u^{(1)}$ and $u^{(2)}$ are respectively the solutions of the following PDEs:
\begin{equation}
\label{eq:u1}
\begin{aligned}
& \frac{\partial u^{(1)} }{\partial t}  + m(t,x)u^{(1)}_x + \frac{1}{2}\nu^2(t,x)u^{(1)}_{xx} + \frac{1-\gamma}{2}\log\frac{2 \pi}{\gamma\sigma^2(t,x)} \\
& + (1-\gamma)\rho\nu(t,x)\sigma(t,x)\frac{\mu(t,x) - r + \rho\nu(t,x)\sigma(t,x)u_x^{(0)}}{\gamma\sigma^2(t,x)}u_x^{(1)} = 0,\ u^{(1)}(T,x)  = 0,
\end{aligned}\end{equation}
and
\begin{equation}
\label{eq:u2}
\begin{aligned}
& \frac{\partial u^{(2)}}{\partial t}  + m(t,x)u^{(2)}_x + \frac{1}{2}\nu^2(t,x)u^{(2)}_{xx} + \frac{1-\gamma}{2\gamma}\rho^2\nu^2(t,x){u_x^{(1)}}^2 \\
& + (1-\gamma)\rho\nu(t,x)\sigma(t,x)\frac{\mu(t,x) - r + \rho\nu(t,x)\sigma(t,x)u_x^{(0)}}{\gamma\sigma^2(t,x)}u_x^{(2)} = 0,\ u^{(2)}(T,x)  = 0.
\end{aligned}
\end{equation}
Moreover,  $\phi^{(2)}$ is the solution to the PDE:
\begin{equation}
\label{eq:phi2}
\begin{aligned}
\frac{\partial \phi^{(2)}}{\partial t} &+ m(t,x)\phi^{(2)}_x + \frac{1}{2}\nu^2(t,x)(\phi^{(2)}_{xx} + 2u^{(0)}_x\phi^{(2)}_x) \\
& + \frac{(1-\gamma)\rho\nu(t,x)}{\gamma\sigma(t,x)}\big[\mu(t,x)-r+\rho\nu(t,x)\sigma(t,x)u^{(0)}_x\big]\phi^{(2)}_x  \\
& -\frac{(1-\gamma)\rho^2\nu^2(t,x)}{2\gamma}{u^{(1)}_x}^2 = 0,\ \phi^{(2)}(T,x) = 0.
\end{aligned}\end{equation}

\begin{lemma}\label{lemma1}
The solution $u$ of \eqref{eq:hjb 4} admits the following Taylor expansion with respect to $\lambda$:
\begin{equation}
\label{TE}		u(t,x) = u^{(0)}(t,x) + \frac{(1-\gamma)(T-t)}{2}\lambda\log\lambda + \lambda u^{(1)}(t,x) +\lambda^2 u^{(2)}(t,x) +  O(\lambda^3).
\end{equation}
\end{lemma}

\proof{Proof.}
The expansion \eqref{TE} is motivated by observing that the $\lambda$ terms in \eqref{eq:hjb 4} are $\frac{1-\gamma}{2}\lambda\log\lambda$ plus $O(\lambda)$.
Plugging $u$ with the expansion into \eqref{thm:hjb solution}, using the equation of $u^{(0)}$ and equating the $\lambda$ and $\lambda^2$ terms in the resulting equation, we easily derive
the equations \eqref{eq:u1} and \eqref{eq:u2}.
\qed
\endproof
\bigskip

The policy that coincides with the mean of the optimal randomized policy is sub-optimal  for the classical Merton problem due to randomization, which can be considered as a loss in the initial wealth.
For each fixed initial time-factor pair $(t,x)$, we define the {\it equivalent relative wealth loss} to be $\Delta=\Delta(t,x)$ such that the investor is indifferent between obtaining the optimal value without the preference for randomization with initial endowment $w(1-\Delta)$ and getting the value of the mean policy, $\operatorname{Mean}(\bm\pi^*)$, with initial endowment $w$. That is, $\Delta$ is such that  $V^{(0)}\big( t,w(1-\Delta),x\big) = V^{(\lambda)}(t,w,x)$. The equivalent relative wealth loss $\Delta$ measures the relative cost that the agent is willing to pay for the pleasure of randomization. 

The following theorem quantifies the bias of the  optimal randomized policy, the relative loss of utility, and the equivalent relative wealth loss, all in terms of $\lambda$ asymptotically.

\begin{theorem}
\label{thm:expansion}
The asymptotic expansion of the mean of the optimal randomized  policy is
$$
\operatorname{Mean}(\bm\pi^*(t,x)) = a^*(t,x) + \lambda\frac{\rho\nu(t,x)}{\gamma\sigma(t,x)}u_x^{(1)}(t,x)+O(\lambda^2),
$$
where $a^*$ is the optimal policy for the classical case. Moreover, we have $$V^{(\lambda)}(t,w,x) = \frac{w^{1-\gamma}}{1-\gamma}\exp\left( u^{(0)}(t,x) + \lambda^2 \phi^{(2)}(t,x) + O(\lambda^3) \right) - \frac{1}{1-\gamma},$$
along with the relative utility loss
\[ \Big|\frac{V^{(\lambda)}(t,w,x)-V^{(0)}(t,w,x) }{V^{(0)}(t,w,x)}\Big| =\lambda^2 |\phi^{(2)}(t,x)|\cdot\Big|1 + \frac{1}{(1-\gamma)V^{(0)}(t,w,x)}\Big|+ O(\lambda^3). \]
Finally, the equivalent relative wealth loss is
\[ \Delta(t,x)
%= 1 - \exp\left\{ \big( \lambda^2 \phi^{(2)}(t,x) + O(\lambda^3) \big)/(1-\gamma)\right\}
= -\frac{\lambda^2 \phi^{(2)}(t,x)}{1-\gamma} + O(\lambda^3). \]
\end{theorem}

\proof{Proof.}
Theorem \ref{thm:hjb solution} along with Lemma \ref{lemma1} imply that the mean of the optimal randomized  policy is expanded as
$$
\begin{aligned}
	a^{(\lambda)}(t,x):=&\frac{\mu(t,x)-r}{\gamma\sigma^2(t,x)} + \frac{\rho\nu(t,x)}{\gamma\sigma(t,x)}\big[ u^{(0)}_x(t,x) + \lambda u^{(1)}_x(t,x) + \lambda^2u^{(2)}_x(t,x) + O(\lambda^3) \big]\\
	= &a^*(t,x) + \lambda\frac{\rho\nu(t,x)}{\gamma\sigma(t,x)}u_x^{(1)}(t,x) +\lambda^2\frac{\rho\nu(t,x)}{\gamma\sigma(t,x)}u_x^{(2)}(t,x)+ O(\lambda^3),
\end{aligned}
$$
where $a^*(t,x) = \frac{\mu(t,x)-r}{\gamma\sigma^2(t,x)} + \frac{\rho\nu(t,x)}{\gamma\sigma(t,x)} u^{(0)}_x(t,x)$ is the optimal policy for  the classical case. %That is, compared with the optimal policy in the classical case, the mean of optimal exploratory policy is biased and differs by $O(\lambda)$ asymptotically, even though the value function differs by $O(\lambda\log\lambda)$. We have better a approximation of policy than the value function in the exploratory framework because the leading error in the value function is mainly caused by the reward from the entropy regularizer, which is expected and independent of the market state. Therefore the extra hedging demand only stays at the level of $O(\lambda)$, which is caused by the entropy of state-dependent exploration.

%To see this, suppose we apply action $a^{\lambda}(t,x) = a^*(t,x) + \lambda\frac{\rho\nu(t,x)}{\gamma\sigma(t,x)}u_x^{(1)}(t,x) +\lambda^2\frac{\rho\nu(t,x)}{\gamma\sigma(t,x)}u_x^{(2)}(t,x) + O(\lambda^3)$ deterministically and we are interested in the expectation of terminal utility
Recall that $V^{(\lambda)}$ is the value function of the classical problem under the deterministic policy $a^{(\lambda)}$. By the Feynman-Kac formula, $V^{(\lambda)}$ satisfies the PDE:
\begin{equation}
	\label{eq:value function under exploratory mean}
	\begin{aligned}
		V^{(\lambda)}_t + & \Big( r + \big(\mu(t,x) - r\big)a^{\lambda}(t,x) \Big)wV^{(\lambda)}_w + \frac{1}{2}\sigma^2(t,x){a^{\lambda}}^2(t,x)w^2V^{(\lambda)}_{ww} \\
		& + m(t,x) V^{(\lambda)}_x + \frac{1}{2}\nu^2(t,x)V^{(\lambda)}_{xx} + \rho\nu(t,x)\sigma(t,x)a^{(\lambda)}(t,x) wV^{(\lambda)}_{wx}  = \beta V^{(\lambda)},
	\end{aligned}
\end{equation}
with the terminal condition $V^{(\lambda)}(T,w,x) =  U(w) = \frac{w^{1-\gamma} - 1}{1-\gamma}$. Conjecturing $V^{(\lambda)}(t,w,x) = \frac{w^{1-\gamma}-1}{1-\gamma}\exp\{\psi^{(\lambda)}(t,x)\} + \varphi^{(\lambda)}(t,x)$ and putting it to \eqref{eq:value function under exploratory mean}, we obtain $-\frac{1}{1-\gamma}\exp\{\psi^{(\lambda)}(t,x)\} + \varphi^{(\lambda)}(t,x) = -\frac{e^{-\beta(T-t)}}{1-\gamma}$, where $\psi^{(\lambda)}$ satisfies
\[
\begin{aligned}
	\psi^{(\lambda)}_t & + m(t,x)\psi^{(\lambda)}_x + \frac{1}{2}\nu^2(t,x)(\psi^{(\lambda)}_{xx} + {\psi^{(\lambda)}_x}^2) + (1-\gamma)\rho\nu(t,x)\sigma(t,x)a^{(\lambda)}(t,x)\psi^{(\lambda)}_x \\
	& + [(1-\gamma)r - \beta] + (1-\gamma)\big[ (\mu(t,x) - r)a^{(\lambda)}(t,x) -\frac{\gamma}{2}\sigma^2(t,x){a^{(\lambda)}}^2(t,x)\big] = 0,\ \psi^{(\lambda)}(T,x) = 0.
\end{aligned}\]

However,  when $\lambda = 0$, $\psi^{(0)} = u^{(0)}$, which motivates us to  expand $\psi^{(\lambda)} = u^{(0)} + \lambda \phi^{(1)} + \lambda^2 \phi^{(2)} + O(\lambda^3)$. Substituting this to the above equation, we deduce that $\phi^{(1)}$ satisfies
\[
\begin{aligned}
	\phi^{(1)}_t & + m(t,x)\phi^{(1)}_x + \frac{1}{2}\nu^2(t,x)(\phi^{(1)}_{xx} + 2u^{(0)}_x\phi^{(1)}_x) \\
	& + \frac{(1-\gamma)\rho\nu(t,x)}{\gamma\sigma(t,x)}\big[\mu(t,x)-r+\rho\nu(t,x)\sigma(t,x)u^{(0)}_x\big]\phi^{(1)}_x = 0,\ \phi^{(1)}(T,x) = 0. \end{aligned}\]
Because $\phi^{(1)} \equiv  0$ is a solution to this linear PDE, it can be easily checked that \eqref{eq:phi2} is the equation satisfied by $\phi^{(2)}$. The proof is complete.
\qed
\endproof
\bigskip

Even though the policy bias is of order $O(\lambda)$, both the relative utility loss and the equivalent relative wealth loss are of order $O(\lambda^2)$. So 
financial and utility losses due to preference for randomization are of higher order of the policy deviation.  

\section{Conclusions}
\label{sec:concl}

This paper aims to address the prevalent appetite  for randomization in a dynamic setting of Merton's problem. We introduce the RPU with entropy functions to represent the preference for stochastic choices, and prove that the optimal policy is Gaussian in a general Markovian incomplete market with CRRA bequest utility. The mean of the Gaussian policy generally differs from the classical Merton solution due to intertemporal hedging demand. An asymptotic expansion in temperature quantifies the deviation of the optimal mean from the classical benchmark and the associated wealth loss as the financial cost of the preference for randomization.

This work opens the gate  to several directions of interesting future research. For example,  dynamic and recursive preference for randomization calls for more solid micro-economics  underpinning. This includes, among others, extending the current framework to include consumption which would allow a reexamination of RPU, and conducting empirical analysis that could help quantify the preference for randomization and assess whether RPU can better explain asset pricing and consumption patterns. It is also interesting to investigate the model-free, RL setting of the current problem, where randomization is out of both necessity (for exploration) and pleasure (for additional utility). In such a setting,  how to distinguish and disentangle these two and how do they interact? Along a different line, for Barberis' model of  optimal exit from  casino gambling \citep{barberis2012model} featuring Kahneman and Tversky's cumulative prospective theory (CPT), \cite{he2017path} and \cite{hu2023casino} show that allowing randomized strategies {\it strictly} improves the optimal CPT value and they attribute this to the non-concavity of the S-shaped utility function in CPT. This suggests that preference for randomization may be implicitly captured by certain non-concave preferences without having to add explicitly a perturbed utility. Technically, it would be valuable to inquire the impact of alternative perturbed functions  such as the Tsallis or R\'enyi entropy and/or alterative bequest utility functions beyond CRRA/CARA, as well as to study RPU for general stochastic control problems.

%Finally, let us remark that this paper focuses on theoretical development without presenting any RL algorithms. However, the theoretical framework and results obtained set the stage for applying immediately the general algorithms developed in \citet{jia2021policy,jia2021pgac,jia2023q}. We leave the details to the interested readers.

\section*{Acknowledgments}
Dai acknowledges the support of Hong Kong GRF (15213422 \& 15217123), The Hong Kong Polytechnic University Research Grants (P0039114, P0042456, P0042708, and P0045342), and the National Natural Science Foundation of China (72432005).  Dong is supported by the National Natural Science Foundation of China (12071333 \& 12101458). Jia acknowledges the support of The Chinese University of Hong Kong start-up grant. Zhou is supported by a start-up grant and the Nie Center for Intelligent Asset Management at Columbia University. His work is also part of a Columbia-CityU/HK collaborative project that is supported by the InnoHK Initiative, The Government of the HKSAR, and the AIFT Lab.

\bibliographystyle{informs2014}
\bibliography{ref}

\bigskip

\begin{APPENDICES}

	\newpage
	
	\section{Different Temperature Schemes}
	
	In this section we discuss two alternative temperature schemes for weighing the entropy utility, and explain the drawbacks of these formulations compared with our recursive formulation. %We will also show that the recursive formulation works for the CARA utility.
	
	\subsection{Constant Temperature}\label{appendix:what if constant temperature}
	With the exploratory dynamics given in \eqref{controlled_system}, if we were to use a constant temperature as in \cite{wang2018exploration} with the objective function
	\[ \E\left[ \int_t^T \lambda \mathcal{H}(\bm\pi_s)\dd s + U(W_T^{\bm{\bm\pi}}) \Big|W_t^{\bm\pi} = w, X_t = x  \right] ,\]
	then the associated HJB equation would be
	\begin{equation}
		\label{eq:hjb 0 equal}
		\begin{aligned}
			& \frac{\partial V}{\partial t} + \sup_{\bm\pi}\Bigg\{ \Big( r + \big(\mu(t,x) - r\big)\operatorname{Mean}(\bm\pi) \Big)wV_w + \frac{1}{2}\sigma^2(t,x)\Big( \operatorname{Mean}(\bm\pi)^2 +  \operatorname{Var}(\bm\pi) \Big)w^2V_{ww} \\
			& + m(t,x) V_x + \frac{1}{2}\nu^2(t,x)V_{xx} + \rho\nu(t,x)\sigma(t,x)\operatorname{Mean}(\bm\pi) wV_{wx} + \lambda \mathcal{H}(\bm\pi) \Bigg\} = 0,
		\end{aligned}
	\end{equation}
	with the terminal condition $V(T,w,x) =  U(w) = \frac{w^{1-\gamma} - 1}{1-\gamma}$.
	
	Similar to \cite{wang2018exploration}, we can solve the maximization problem on the left-hand side of \eqref{eq:hjb 0 equal} and apply the verification theorem to conclude that the optimal policy is a normal distribution with mean
	\[\frac{\big( \mu(t,x) - r \big)V_w}{-\sigma^2(t,x)wV_{ww}} + \frac{\rho \nu(t,x) V_{wx}}{-\sigma(t,x)w V_{ww}}, \]
	and variance
	\[ \frac{\lambda}{-\sigma^2(t,x)w^2V_{ww}}. \]
	Plugging the above into \eqref{eq:hjb 0 equal}, the equation  becomes
	\begin{equation}
		\label{eq:hjb 1 equal}
		\begin{aligned}
			& \frac{\partial V}{\partial t} + rwV_w + m(t,x)V_x + \frac{1}{2}\nu^2(t,x)V_{xx} + \frac{\lambda}{2}\log 2\pi \\
			& - \frac{\bigg( \big( \mu(t,x) - r \big)V_w + \rho \nu(t,x)\sigma(t,x)V_{wx}   \bigg)^2}{2\sigma^2(t,x)V_{ww}} +\frac{\lambda}{2}\log\frac{\lambda}{-\sigma^2(t,x)w^2V_{ww}}=0.
		\end{aligned}
	\end{equation}
	
	To our best knowledge, this PDE generally admits neither a separable nor a closed-form solution due to a lack of the homothetic property. As a result, analytical forms of the optimal value function and optimal policy are both unavailable, making it hard to carry out further theoretical analysis such as a comparative study and sensitivity analysis.
	
	\subsection{Wealth-Dependent Temperature}\label{sec: wealth-weighted}
	In order to obtain a simpler HJB equation, we could consider a wealth-dependent temperature parameter. For example, for CRRA utility, we could take $\lambda(t,w,x) = \lambda w^{1-\gamma}$ where $\gamma\neq 1$. The derivations of \eqref{eq:hjb 0 equal} and \eqref{eq:hjb 1 equal} are similar  by replacing $\lambda$ with $\lambda w^{1-\gamma}$. With this wealth-dependent weighting scheme, the problem becomes homothetic in wealth  with degree $1-\gamma$; hence the value function admits the form $V(t,w,x) = \frac{w^{1-\gamma}v(t,x)}{1-\gamma} + g(t)$, where $g$ and $v$ satisfy
	\[ \begin{aligned}
		&(1-\gamma)g' - (1-\gamma)g = \frac{\partial v}{\partial t} + mv_x+\frac{\nu^2}{2}v_{xx},\ g(T,x) = -\frac{1}{1-\gamma}, \\
		& \frac{\partial v}{\partial t} + (1-\gamma)r v + mv_x + \frac{\nu^2}{2}v_{xx} + \lambda \frac{1-\gamma}{2}\log\frac{2\pi}{\gamma\sigma^2v} + \frac{(1-\gamma)\big( (\mu-r)v + \rho\nu\sigma v_x \big)^2}{2\gamma\sigma^2v}  , \ v(T,x) = 1.
	\end{aligned}\]
	Consequently, the optimal policy is a normal distribution with mean $\frac{\mu(t,x) - r }{\gamma\sigma^2(t,x)} + \frac{\rho \nu(t,x) v_{x}(t,x)}{\gamma\sigma(t,x)v(t,x)}$ and variance $\frac{\lambda}{\gamma \sigma^2(t,x) v(t,x)}$.
	
	The major difference between this formulation and the recursive one is that there is an extra term $v(t,x)$ in the denominator of the optimal variance of the former, which may result in an arbitrarily large randomization  variance within a finite time period and consequently the non-existence of an optimal policy. For example, in the Black-Scholes case ($m, \nu \equiv 0$), $v$ is independent of $x$ and satisfies an ODE whose solution will reach  $0$ with some choice of the coefficients. Specifically, $v(t) = \varphi(T - t)$, where $\varphi$ satisfies an ODE:
	\begin{equation}
		\label{eq:ansatz ode crra}
		\varphi'(\tau) = (1-\gamma)\Big( (r+\frac{(\mu-r)^2}{2\gamma\sigma^2})\varphi(\tau) + \frac{\lambda}{2}\log\frac{2\pi\lambda}{\gamma\sigma^2} -\frac{\lambda}{2}\log \varphi(\tau) \Big),\ \varphi(0) = 1.
	\end{equation}
	We have the following theorem.
	\begin{theorem}
		\label{thm:utility maximization explore}
		If $\gamma>1$ and $r+\frac{(\mu-r)^2}{2\gamma\sigma^2} + \frac{\lambda}{2}\log\frac{2\pi\lambda}{\gamma\sigma^2}< 0$, then the solution to \eqref{eq:ansatz ode crra} reaches $0$ in a finite time.
	\end{theorem}
\proof{Proof.}
Consider an ODE $y' = F(y)$ with the initial condition $y(0)=1$, where $F(x) = ax + b -c\log x$ with some constants $a,b$ and $c$. Let  $y(\tau) = e^{-z(\tau)}$. Then $z$ satisfies
\begin{equation}\label{ode:explode}
z' = -e^{z}F(e^{-z})=-a - be^z - cze^{z},\ z(0)=0.
\end{equation}
If $F(x)<0$ for  any $x\in(0,1]$ or $z$ is increasing,  then applying  the Osgood test (c.f., \citealt{ceballos2010generalization}) we conclude that the solution of  \eqref{ode:explode} explodes at time $\tau_e = -\int_0^{\infty} \frac{1}{a+be^{z}+cze^{z}}\dd z$.
The desired result now follows from the fact that $y\to 0$ is equivalent to  $z\to \infty$.
\qed
\endproof

	\section{CARA Utility}\label{sect: cara}
	The RPU also works for the CARA utility, for which we use the dollar amount invested in the risky asset as the control (portfolio) variable, and denote by $\bm\pi$  the corresponding probability-density-valued control. The wealth process under $\bm\pi$ is
	\begin{equation}
		\label{controlled_system cara}
		\dd W^{\bm\pi}_t  = \left[rW^{\bm\pi}_t+(\mu(t,X_t)-r)\operatorname{Mean}(\bm\pi_t)\right]\dd t +\sigma(t,X_t)\left[ \operatorname{Mean}(\bm\pi_t) \dd B_t+\sqrt{ \operatorname{Var}(\bm\pi_t) }\dd \bar{B}_t\right].
	\end{equation}
	Consider the following regularized objective function:
	\begin{equation}
		\label{objective_functional cara}
		J^{\bm\pi}_t:=\mathbb E\left[ \int_t^T  -\lambda J_s^{\bm\pi}\mathcal{H}(\bm\pi_s) \dd s  +  U(W_T^{\bm\pi}) \Big| \f_t \right] ,
	\end{equation}
	where $U(x) = -\frac{1}{\gamma} e^{-\gamma x}$. Under this recursive weighting scheme, the optimal policy is given by
	\[ \bm\pi^*(t,x) = \mathcal{N}\bigg( \big( \frac{\mu(t,x) - r}{\gamma \sigma^2(t,x)} - \frac{\rho\nu(t,x)}{\sigma(t,x)}u_x(t,x)\big) e^{-r(T-t)}, \frac{\lambda}{\gamma^2\sigma^2(t,x)}e^{-2r(T-t)} \bigg) ,\]
	where $u$ satisfies
	\begin{align*}
		& u_t + m(t,x)u_x + \frac{1}{2}\nu^2(t,x)(-\gamma u_x^2 + u_{xx})  + \frac{\gamma\sigma^2(t,x)}{2}\left(  \frac{\mu(t,x) - r}{\gamma \sigma^2(t,x)} - \frac{\rho\nu(t,x)}{\sigma(t,x)}u_x \right)^2 \\
		&+ \frac{\lambda}{2\gamma}\log\frac{2\pi\lambda}{\gamma^2\sigma^2(t,x)e^{2r(T-t)}}  = 0,\ \ \ u(T,x) = 0,
	\end{align*}
	and the associated optimal value function is $V(t,w,x) = -\frac{1}{\gamma}\exp( -\gamma u(t,x) - \gamma e^{r(T-t)}w )$. Hence we can develop a theory  parallel to the CRRA utility. Details are left to the interested readers.
	
	\section{%Choice of $\lambda$ and Relations with Classical Counterpart:
		A BSDE Perspective}\label{app:bsde}
	The optimal portfolio processes with parameter $\lambda$ in the RPU problem can also be characterized by a BSDE.
	\begin{theorem}
		\label{thm:optimal bsde}
		Suppose the following quadratic BSDE admits a unique solution
		\begin{equation}
			\label{eq: optimal bsde}
			\left\{  \begin{aligned}
				\dd Y_t^{*(\lambda)} = & -\Bigg\{ (1-\gamma)\Big[ \frac{\lambda}{2}\log\frac{2 \pi \lambda}{\gamma\sigma_t^2} + r + \frac{(\mu_t - r + \rho\sigma_t Z^{*(\lambda)}_t)^2}{2\gamma\sigma_t^2} \Big]  + \frac{1}{2}{Z_t^{*(\lambda)}}^2\Bigg\} \dd t + Z_t^{*(\lambda)} \dd B_t^X,\\
				Y^{*(\lambda)}_T = & 0.
			\end{aligned}\right.
		\end{equation}
		Then the optimal randomized  control is given by
		\[ \bm\pi^{*(\lambda)}_t = \mathcal{N}\Big( \frac{\mu_t - r + \rho\sigma_t Z^{*(\lambda)}_t}{\gamma\sigma^2_t}, \frac{\lambda}{\gamma \sigma_t^2}  \Big)   . \]
	\end{theorem}
	
	The solution to BSDE \eqref{eq: optimal bsde} corresponds to the PDE \eqref{eq:hjb 4} in Theorem \ref{thm:hjb solution}.
	The process $Y^{*(\lambda)}$ can be interpreted as an auxiliary process in the martingale duality theory, which stipulates that
	\[ \left( V^{\bm\pi}_t + \int_0^t  \lambda \mathcal{H}(\bm\pi_s)\big[(1-\gamma)V^{\bm\pi}_s + 1\big]\dd s \right) e^{Y_t^{*(\lambda)}} \]
	is a supermartingale for any portfolio control $\bm\pi_t$ and a martingale when $\bm\pi_t = \bm\pi^{*(\lambda)}_t$, where $V^{\bm\pi}_t = \frac{{w_t^{\bm\pi}}^{1-\gamma}}{1-\gamma}e^{Y_t^{*(\lambda)}} -  \frac{1}{1-\gamma}$ is the value function evaluated along the optimal wealth  trajectories. Moreover,
	when $\lambda = 0$, the above result reduces to the classical one which was first derived by \cite{hu2005utility}.
	
	The BSDE  \eqref{eq: optimal bsde} offers important insights  about the connection between the classical problem ($\lambda=0$) and the RPU ($\lambda > 0$), some of them consistent with the unbiasedness discussion  in Subsection \ref{sec:unbi}. First, the only difference between the two problems is the extra term, $\frac{(1-\gamma)\lambda}{2}\log\frac{2\pi\lambda}{\gamma\sigma_t^2}$, in the driver of \eqref{eq: optimal bsde}. This term becomes a deterministic function of $t$ when $\gamma=1$ (log-utility) or when $\sigma_t$ is deterministic, in which case $Z^{*(\lambda)} = Z^{*(0)}$.\footnote{This is because, in general,  if $(Y_t,Z_t)$ satisfies a BSDE $\dd Y_t = -f(t,X_t,Z_t)\dd t + Z_t\dd B^X_t$ with $Y_T = F(X_T)$,  then $(Y_t+C(t),Z_t)$ solves the BSDE $\dd \tilde{Y}_t = -\big[c(t)+f(t,X_t,\tilde{Z}_t)\big]\dd t + \tilde{Z}_t\dd B^X_t$ with $\tilde{Y}_T = F(X_T)$, where $c$ is a given deterministic function of $t$ and $C(t) = \int_t^T c(s)\dd s$. Hence $Z=\tilde Z$ follows from the uniqueness of solution.} Hence the optimal solution of the RPU problem has a mean that coincides with the classical solution, or the former is unbiased. Incidentally, this is consistent with an earlier result of \cite{ddj2020calibrating} in the mean-variance analysis for the log utility. Second, when $\rho=0$, i.e., the market factors evolve independently from the stock price, the optimal randomized  policy is also unbiased, even if $Z^{*(\lambda)} \neq Z^{*(0)}$ in general. This is intuitive because  hedging is not necessary in this case.  A special case is when the market factor $X$ is deterministic, in which case the source of randomness $B^X_t$ vanishes leading to  $Z^{*(\lambda)} = Z^{*(0)} = 0$ and hence the optimal randomization is unbiased.
	Otherwise, optimal  policies are in general biased  in stochastic volatility models when $\sigma_t$ is a function of the stochastic factor process $X_t$ and the power of the utility $\gamma \neq 1$.
	%\section{Expressions for the Solutions to ODEs}
	%\label{appendix:expression}
	%In \eqref{eq:classical solution ricatti},
	%\[ \psi_0 = -\frac{\sqrt{\iota^2 \gamma - (1-\gamma) \delta\bar\nu(\delta\nu + 2\iota\rho)}}{\sqrt{\gamma}}, \ \psi_1 = \frac{(1-\gamma)\delta^2}{\bar{\nu}^2[\rho^2 + \gamma(1-\rho^2)]} ,\]
	%\[ \psi_{2,3} = \frac{\iota \gamma - (1-\gamma)\delta\bar\nu\rho \pm \sqrt{\gamma} \sqrt{\iota^2 \gamma - (1-\gamma) \delta\bar\nu(\delta\nu + 2\iota\rho)} }{\bar{\nu}^2[\rho^2 + \gamma(1-\rho^2)]} ,\]
	%\[ \psi_4 = (1-\gamma)r -\beta + \iota\bar x\psi_3,\ \psi_5 = \frac{2\gamma\iota\bar x}{\bar{\nu}^2[\rho^2 + \gamma(1-\rho^2)]}  .\]
	%
	%In \eqref{eq:explode ode},
	%\[ \psi_0 = \frac{\iota \gamma - (1-\gamma)\delta\bar\nu\rho }{\bar{\nu}^2[\rho^2 + \gamma(1-\rho^2)]},\ \psi_1 = \frac{\sqrt{\gamma}\sqrt{-\iota^2 \gamma + (1-\gamma) \delta\bar\nu(\delta\nu + 2\iota\rho)}}{\bar{\nu}^2[\rho^2 + \gamma(1-\rho^2)]} ,\ \psi_2 = \frac{\sqrt{-\iota^2 \gamma + (1-\gamma) \delta\bar\nu(\delta\nu + 2\iota\rho)}}{2\sqrt{\gamma}},\]
	%\[\psi_4 = (1-\gamma)r - \beta -\iota \bar{x} \psi_0,\ \psi_5 = -\frac{2\gamma\iota\bar{x}}{\bar{\nu}^2[\rho^2 + \gamma(1-\rho^2)]}.  \]

\end{APPENDICES}

\end{document}